\documentclass[twocolumn, journal]{ieeetran}

\usepackage{graphics} 
\usepackage{epsfig} 

\usepackage{amsmath,amssymb,amsthm,amsfonts} 
\usepackage{bbm}
\usepackage{tikz, epic,eepic}
\usepackage{caption}
\usepackage{subcaption}
\usetikzlibrary{shapes,arrows}
\usetikzlibrary{automata}
\usepackage{epsfig}
\usepackage{graphicx}
\usepackage{caption}
\usepackage{subcaption}
\usepackage{algorithm}
\usepackage[noend]{algpseudocode}
\usepackage{multirow,array}
\usepackage{dsfont}
\usepackage{arydshln}
\usepackage{bm}
\usepackage{cite}
\usepackage{epstopdf}

\let\labelindent\relax
\usepackage{enumitem}

\newtheorem{theorem}{Theorem}
\newtheorem{fact}{Fact}
\newtheorem{claim}{Claim}

\newtheorem{corollary}{Corollary}
\newtheorem{lemma}{Lemma}
\newtheorem{definition}{Definition}

{\itshape}{\rmfamily}

\newtheorem{remark}{Remark}

\newcommand{\argmax}[1]{\underset{#1}
{\operatorname{arg}\,\operatorname{max}}\;}
\newcommand{\argmin}[1]{\underset{#1}
{\operatorname{arg}\,\operatorname{min}}\;}
\def\mcal{\mathcal}
\def\mbb{\mathbb}
\def\JB{J_{\text{b}}}
\def\JBstar{J_{\text{b}}^*}
\def\JF{J_{\text{f}}}
\def\JFstar{J_{\text{f}}^*}
\def\RB{R_{\text{b}}}
\def\RF{R_{\text{f}}}
\def\RBstar{R_{\text{b}}^*}
\def\RFstar{R_{\text{f}}^*}
\def\gammaB{\gamma_{\text{b}}}
\def\gammaF{\gamma_{\text{f}}}

\def\alphasys{\alpha_{\text{sys}}}

\title{\LARGE \bf
A risk-security tradeoff in graphical coordination games
}

\author{Keith Paarporn, Mahnoosh Alizadeh, and Jason R. Marden 
\thanks{K. Paarporn ({\tt\small kpaarporn@ucsb.edu}), M. Alizadeh ({\tt\small alizadeh@ucsb.edu}), and J. R. Marden ({\tt\small jrmarden@ece.ucsb.edu}) are with the Department of Electrical and Computer Engineering at the University of California, Santa Barbara, CA.  } \thanks{A preliminary version of this paper has been submitted for conference publication (CDC 2019) and can be found at {\tt\footnotesize https://ece.ucsb.edu/\~jrmarden/files/MardenC9.pdf}. The analysis of randomized operator strategies extends the previous results. Complete proofs are also provided here. This work is supported by UCOP Grant LFR-18-548175, ONR grant \#N00014-17-1-2060, and NSF grant \#ECCS-1638214.}
}

\begin{document}

\maketitle
\thispagestyle{empty}
\pagestyle{empty}

\begin{abstract}
	A system relying on the collective behavior of decision-makers can be vulnerable to a variety of adversarial attacks. How well can a system operator protect performance in the face of these risks? We frame this question in the context of graphical coordination games, where the agents in a network choose among two conventions and derive benefits from coordinating neighbors, and system performance is measured in terms of the agents' welfare. In this paper, we assess an operator's ability to mitigate two types of adversarial attacks - 1) {broad} attacks, where the adversary incentivizes all agents in the network and 2) {focused} attacks, where the adversary can force a selected subset of the agents to commit to a prescribed convention. As a mitigation strategy, the system operator can implement a class of distributed algorithms that govern the agents' decision-making process. Our main contribution characterizes the operator's fundamental trade-off between security against worst-case {broad} attacks and vulnerability from {focused} attacks. We show that this tradeoff significantly improves when the operator selects a decision-making process at random. Our work highlights the design challenges a system operator faces in maintaining resilience of networked distributed systems.

\end{abstract}
\section{Introduction}\label{sec:intro}

Networked distributed  systems typically operate without centralized planning or control, and instead rely on local interactions and communication between the comprising agents. These systems arise in a variety of engineering applications such as teams of mobile robots and sensor networks \cite{Cortes_2006,Mesbahi_2010,Akyildiz_2002}. They are also prevalent in social dynamics \cite{Young_2001,Montanari_2010} and biological populations \cite{West_2007}.  




The transition from a centralized to a distributed architecture may leave a system vulnerable to a variety of adversarial attacks. An adversary may be able to manipulate the decision-making processes of the agents. Such dynamical perturbations can potentially lead to unwanted outcomes. For example in social networks, individual opinions can be shaped from  external information sources, resulting in a polarized society \cite{DelVicario_2016,Mao_2018}. When feasible, a system operator takes  measures to mitigate adversarial influences. The literature on cyber-physical system security studies many aspects of this interplay. For instance, optimal controllers are designed to mitigate denial-of-service,  estimation, and deception attacks \cite{Amin_2009, Amin_2010,Fawzi_2011,Pasqualetti_2013, Sundaram_2018}. 


This paper investigates security measures that a system operator can take against adversarial influences when the underlying system is a  graphical coordination game \cite{Montanari_2010,Auletta_2012}, where agents in a network decide between two choices, $x$ or $y$. One may think of these choices as two competing products, e.g. iPhone vs Android, two conflicting social norms, or two opposing political parties. Each agent derives a positive benefit from  interactions with coordinating neighbors, and zero benefits from mis-coordinating ones.  The system's efficiency is defined by the ratio of total benefits of all agents to the maximal attainable benefits over all configurations of choices. 

The goal of the system operator is to design a local decision-making rule for each agent in the system so that the emergent collective behavior optimizes system efficiency. One algorithm that achieves this goal is known as log-linear learning \cite{Blume_1995,Marden_2012,Tatarenko_2014}. More formally, the agents follow a ``perturbed" best reply dynamics where the agents' local objectives are precisely equal to their local welfare. We seek to address the question of whether this particular algorithm is robust to adversarial influences. That is, does this algorithm preserve system efficiency when the agents' decision-making processes are manipulated by an adversary? If not, can the operator alter the agents' local objectives to mitigate such attacks?

We consider two adversarial attack models - \emph{broad} and \emph{focused} attacks. In {broad} attacks, the adversary incentivizes every agent in the network (hence broad) with a convention, influencing their decision-making process.  This could depict distributing  political ads with the intention of polarizing voters. In {focused} attacks, the adversary targets a specific set of  agents in the network, forcing them to commit to $x$ or $y$. These targeted, or fixed agents consequently do not update their choices over time but still influence the decisions of others. For instance, they could portray loyal consumers of a brand or product, or staunch supporters of a political party.  Fixed agents and their effects on system performance have been extensively studied in the context of opinion dynamics and optimization algorithms \cite{Acemoglu_2010,Ghaderi_2014,Sundaram_2018}. 


The first contribution of this paper is a characterization of worst-case risk metrics from both adversarial attacks as a function of the operator's algorithm design parameter (Section \ref{sec:models}). We define risk in this paper as the system's distance to optimal efficiency. By worst-case here we mean the maximum risk among all connected network topologies subject to any admissible adversarial attack. Hence, our analysis identifies the network topologies on which worst-case risks are attained (Section \ref{sec:analysis}).  We extend this analysis to randomized operator strategies (Sections \ref{sec:mixing}, \ref{sec:analysis2}). 

The second contribution of this paper answers the question ``if the operator succeeds in protecting the system from one type of attack, how vulnerable does it leave the system to the other?" We identify a fundamental tradeoff between security against {broad} attacks and risks from {focused} attacks.  We then show randomized operator strategies significantly improves the set of attainable risk levels and their associated tradeoffs (Section \ref{sec:mixing}).

By characterizing this interplay, we contribute to previous work that studied the impact of adversarial influence in graphical coordination games \cite{Borowski_2015,Brown_2018,Canty_2018}. These works analyze worst-case damages that can be inflicted by varying degrees of adversarial sophistication and intelligence in the absence of a system operator. However, these results were derived only in specific graph structures, namely ring graphs, whereas our analysis considers adversarial influence in any graph topology.

\section{Preliminaries}\label{sec:prelim}

\subsection{Graphical coordination games}
A graphical coordination game is played between a set of agents $\mcal{N}=\{1,\ldots,N\}$ over a connected undirected network $G = (\mcal{N},\mcal{E})$ with node set $\mcal{N}$ and edge set $\mcal{E} \subset \mcal{N}\times\mcal{N}$. Agent $i$'s set of neighbors is written as $\mcal{N}_i = \{j : (i,j) \in \mcal{E} \}$. Each agent $i$ selects a choice $a_i$ from its action set $\mcal{A}_i = \{x,y\}$.  The choices of all the agents constitutes an action profile $a = (a_1,\ldots,a_N)$, and we denote the set of all action profiles as $\mcal{A} = \Pi_{i=1}^N \mcal{A}_i$. The local interaction between two agents $(i,j) \in \mcal{E}$ is based on a $2\times 2$ matrix game, described by the payoff matrix $V : \{x,y\}^2 \rightarrow \mbb{R}$,
\begin{equation}\label{eq:base_game}
    \begin{tabular}{cc|c|c|}
      & \multicolumn{1}{c}{} & \multicolumn{2}{c}{Player $j$}\\
      & \multicolumn{1}{c}{} & \multicolumn{1}{c}{$x$}  & \multicolumn{1}{c}{$y$} \\\cline{3-4}
      \multirow{2}*{Player $i$}  & $x$ & $1+\alpha_{\text{sys}},1+\alpha_{\text{sys}}$ & $0,0$ \\\cline{3-4}
      & $y$ & $0,0$ & $1,1$ \\\cline{3-4}
    \end{tabular}
\end{equation}
where $\alpha_{\text{sys}} > 0$ is the system \emph{payoff gain}. It indicates that $x$ is an inherently superior product over $y$ when users coordinate. Here, agents would rather coordinate than not, but prefer to coordinate on $x$. Agent $i$'s \emph{benefit} is the sum of payoffs derived from playing the game \eqref{eq:base_game} with each of its network neighbors:
\begin{equation}\label{eq:utility}
	W_i(a_i,a_{-i}) := \sum_{j\in \mcal{N}_i} V(a_i,a_j) \ .
\end{equation}
 A measure of system \emph{welfare} defined over $\mcal{A}$ is
\begin{equation}\label{eq:welfare}
	W(a) := \sum_{i=1}^N W_i(a_i,a_{-i}),
\end{equation}
which is simply the sum of all agent benefits. The \emph{system efficiency} for action profile $a\in\mcal{A}$ is defined as  
\begin{equation}\label{eq:efficiency}
	\frac{W(a)}{\max_{a'\in\mcal{A}} W(a')}.
\end{equation}
For $\mcal{A} = \{x,y\}^N$, the all-$x$ profile $\vec{x}$ maximizes welfare. This does not necessarily hold for arbitrary action spaces.  

\subsection{Log-linear learning algorithm}
Log-linear learning is a distributed stochastic algorithm governing how players' decisions evolve over time \cite{Marden_2012,Auletta_2012,Blume_1995}. It may be applied to any instance of a game with each player having a well-defined local utility function $U_i: \mcal{A} \rightarrow \mbb{R}$ over a set of action profiles $\mcal{A}$ with an underlying interaction graph $G$. That is, agent $i$'s local utility is a function of its action $a_i$ and actions of its neighbors in $G$. 

Agents update their decisions $a(t) \in \mcal{A}$ over discrete time steps $t=0,1,\ldots$. Assume $a(0)$ is arbitrarily determined. For step $t \geq 1$, one agent $i$ is selected uniformly at random from the population. It updates its action to $a_i(t)=z \in \mcal{A}_i$ with probability 
\begin{equation}
	\frac{\text{exp}(\beta U_i(z,a_{-i}(t-1))}{\sum_{z' \in \mcal{A}_i} \text{exp}(\beta U_i(z',a_{-i}(t-1))},
\end{equation}
where $\beta > 0$ is the rationality parameter. All other agents repeat their previous actions: $a_{-i}(t) = a_{-i}(t-1)$.  For large values of $\beta$, $i$ selects a best-response to the previous actions of others with high probability, and for values of $\beta$ near zero, $i$ randomizes among its actions $\mcal{A}_i$ uniformly at random. This induces an irreducible Markov chain over the action space $\mcal{A}$, with a unique stationary distribution $\pi_\beta \in \Delta(\mcal{A})$. The \emph{stochastically stable states} (SSS) $a \in \mcal{A}$ are the action profiles contained in the support of the stationary distribution in the high rationality limit: they satisfy $\pi(a) = \lim_{\beta\rightarrow\infty} \pi_\beta(a) > 0$. Such a limiting distribution exists and is unique \cite{Young_1993,Blume_1995,Foster_1990}. We write the set of stochastically stable states as 
\begin{equation}\label{eq:LLL}
	\text{LLL}(\mcal{A},\{U_i\}_{i\in\mcal{N}};G).
\end{equation}

For graphical coordination games, the log-linear learning algorithm specified by the action set $\mcal{A} = \{x,y\}^N$ and utilities $\{W_i\}_{i\in\mcal{N}}$ selects the welfare-maximizing profile $\vec{x}$ as the stochastically stable state irrespective of the graph topology $G$. This can be shown using standard potential game arguments \cite{Marden_2012} (we provide these details in Section \ref{sec:analysis}). That is, $\vec{x} = \text{LLL}(\mcal{A},\{W_i\}_{i\in\mcal{N}};G)$ for all $G \in \mcal{G}_N$, where $\mcal{G}_N$ is the set of all connected undirected graphs on $N$ nodes. 

However, if an adversary is able to manipulate the agents' local decision-making rules, this statement may no longer hold true. A system operator may be able to alter the agents' local utility functions with the goal of mitigating the loss of system efficiency in the presence of adversarial influences. In particular, we consider the class of local utility functions $\{U_i^\alpha\}_{i\in\mcal{N}}$ parameterized by $\alpha>0$. Specifically, $U_i^\alpha$ takes the same form as the benefit function \eqref{eq:utility} where $\alpha_{\text{sys}}$ is replaced with a \emph{perceived gain} $\alpha$ that is under the operator's control. We next introduce models of adversarial attacks in graphical coordination games. We then evaluate the performance of this class of distributed algorithms in the face of adversarial attacks. 

\begin{figure*}[t]
	\centering
	\begin{subfigure}{.45\textwidth}
		\centering
		\includegraphics[scale=1.25]{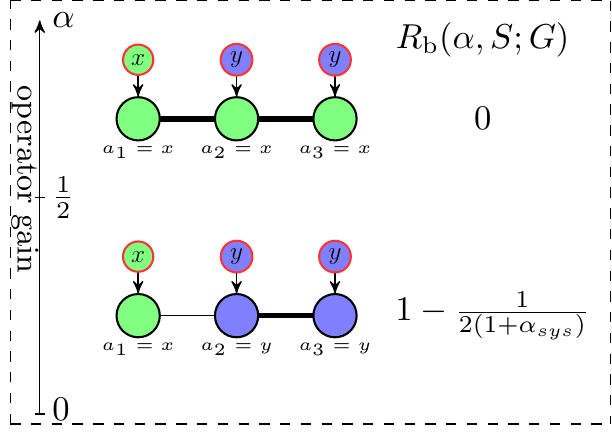}
		\caption{}
		\label{fig:GA_example}
	\end{subfigure}
	\begin{subfigure}{.45\textwidth}
		\centering
		\includegraphics[scale=1.25]{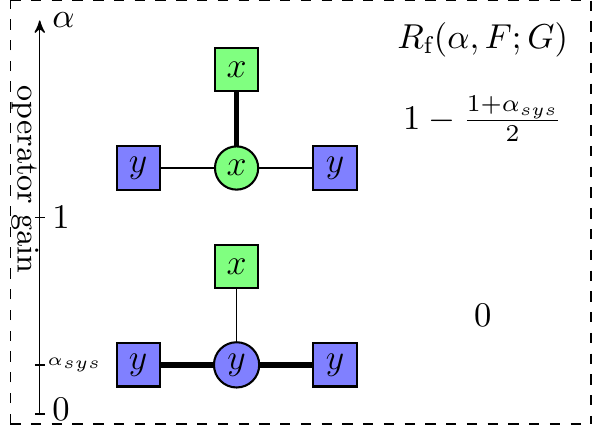}
		\caption{}
		\label{fig:GE_example}
	\end{subfigure}
	\caption{\small (Left) An example three-node line network under a {broad} adversarial attack. The imposter nodes are depicted as the labelled smaller circles and agents in the network are the bigger circles. The color of each circle indicates the node's action - green for $x$, blue for $y$. In this example, maximum welfare is $\max_{a\in\mcal{A}} W(a) = 4(1+\alphasys)$, achieved when all three agents play $x$. The adversary's target set $S$ attaches an $x$-imposter to node 1 and $y$-imposters to nodes $2$ and $3$.  For operator gains $\alpha \leq \frac{1}{2}$, $a=(a_1,a_2,a_3)=(x,y,y)$ is the welfare-minimizing SSS, i.e. it satisfies $a=\argmin{a\in\text{LLL}(\mcal{A},\alpha,S;G)} W(a)$. This gives a risk of $\RB(\alpha,S;G) = 1 - \frac{1}{2(1+\alpha_{\text{sys}}) } $.  For $\alpha > \frac{1}{2}$, the welfare-minimizing SSS is $(x,x,x)$. This gives optimal efficiency, i.e. a risk of $0$. (Right) An example of a four node star network under a {focused} attack where a subset $F$ of three nodes are targeted to be fixed (squares). Only the center node is unfixed. In this example, the maximum welfare is $\max_{a\in\mcal{A}_F} W(a) = 4$, achieved when the center plays $y$. This is because the alternative action (when center plays $x$) gives the suboptimal welfare $2(1+\alphasys) < 4$ due to $\alpha_{\text{sys}} < 1$. For operator gains $\alpha < 1$, the center node plays $y$ in the SSS. This yields optimal efficiency, i.e. the risk is $\RF(\alpha,F;G) = 0$.  For $\alpha \geq 1$, the center node plays $x$, giving a risk of $\RF(\alpha,F;G) =1 - \frac{1 + \alpha_{\text{sys}}}{2}$. The methods to calculate stochastically stable states under both types of attacks follow standard potential game arguments and are detailed in Section \ref{sec:analysis}.} 
\end{figure*}

\section{Models of adversarial influence}\label{sec:models}

In this section, we outline two models of adversarial attacks in graphical coordination games - \emph{broad} and \emph{focused} attacks. The system operator  specifies the local utility functions $\{U_i^\alpha\}$ that govern the log-linear learning algorithm by selecting the perceived payoff gain $\alpha>0$. Our goal is to assess the performance of this range of algorithms on two corresponding worst-case risk metrics, which we define and characterize. We then identify fundamental tradeoff relations between these two risk metrics. 

\subsection{{Broad} attacks and worst-case risk metric}
We consider a scenario where the system is subject to broad adversarial attacks. For each agent in the network, the adversary attaches a single imposter node that acts as a  neighbor that always plays $x$ or $y$. These nodes are not members of the network but affect the decision making of agents that are. Let $S_x \subseteq \mcal{N}$ ($S_y$) be the set of agents targeted with an imposter $x$ ($y$) node.  We call the \emph{target set} $S = (S_x,S_y)$. Any target set satisfies $S_x \cap S_y = \varnothing$ and $S_x \cup S_y = \mcal{N}$. We call $\mcal{T}(G)$ the set of all possible target sets $S$ on the graph $G$. Given $\alpha>0$, the agents' \emph{perceived} utilities are 
\begin{equation}\label{eq:utilities2}
	\tilde{U}_i^{\alpha}(a_i,a_{-i}) := 
	\begin{cases} 
		U_i^{\alpha}(a_i,a_{-i}) + \mathds{1}(a_i=y) \quad & i \in S_y \\  
		U_i^{\alpha}(a_i,a_{-i}) + (1+\alpha)\mathds{1}(a_i=x)  & i \in S_x 
	\end{cases}.
\end{equation}
%

\begin{figure*}[t]
	\centering
	\begin{subfigure}{.48\textwidth}
		\centering
		\includegraphics[scale=1.1]{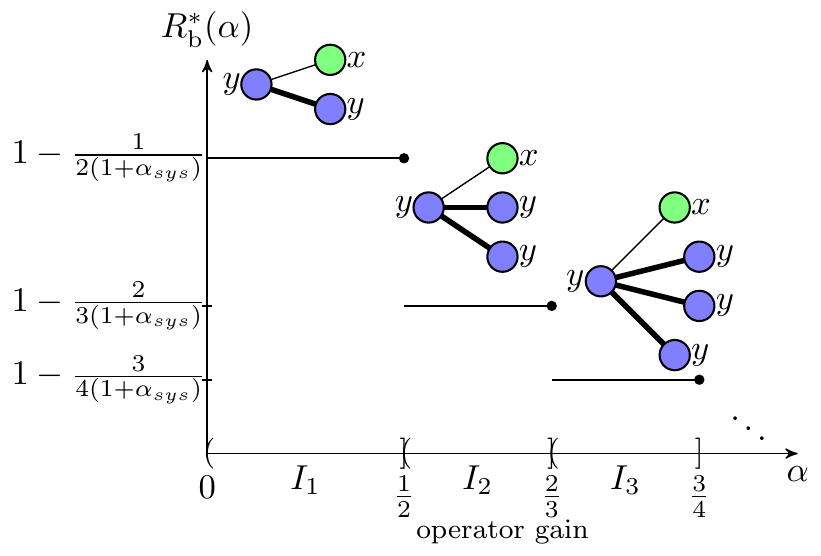}
		\caption{}
		\label{fig:RAstar}
	\end{subfigure}  
	\begin{subfigure}{.48\textwidth}
		\centering
		\includegraphics[scale=1.1]{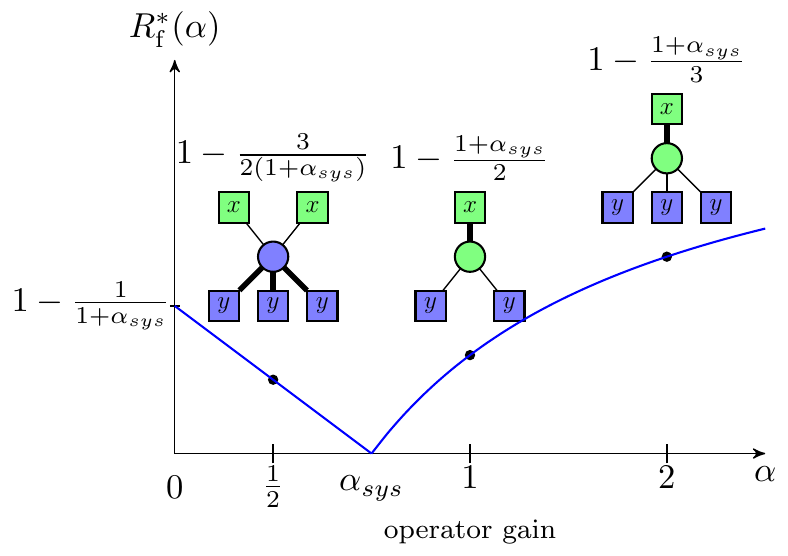}
		\caption{}
		\label{fig:REstar}
	\end{subfigure}
	\caption{\small (a) The worst-case risk from {broad} attacks $\RBstar(\alpha)$  \eqref{eq:RA_WC} is a piecewise constant function defined over countably infinite half-open intervals. The graphs and their corresponding target set which attain each level of worst-case broad risk are illustrated for $\alpha < 1$. Here, the $x,y$ labels indicate the type of imposter influence on the agents (circles) in the network, and the color of the circles depict the action played in the welfare-minimizing SSS (green=$x$, blue=$y$). If $\alpha \in I_k$, $k=1,2,\ldots$(recall \eqref{eq:Ik}), the worst-case risk is achieved on a star graph of $k+2$ nodes where all nodes but one are targeted with a $y$ imposter. The one leaf node has an $x$ imposter attached, giving a single miscoordinating link in the network.  (b) The worst-case risk from {focused} attacks $\RFstar(\alpha)$ \eqref{eq:RE_WC}. The graphs and their corresponding fixed sets which attain the worst-case focused risks are illustrated for $\alpha = \frac{1}{2}, 1$, and $2$. The nodes' color represents the worst-case SSS at $\alpha$ (blue $=y$, green $=x$). The targeted fixed agents are represented as squares and the unfixed agents as circles. Here $\frac{1}{2}< \alphasys < 1$. The proofs establishing all worst-case graphs are detailed in Section \ref{sec:analysis}.} 
	\label{fig:risk_curves}
\end{figure*}

In the notation of \eqref{eq:LLL}, the set of stochastically stable states is written $\text{LLL}(\mcal{A},\{\tilde{U}_i^\alpha\}_{i\in\mcal{N}};G)$.
However for more specificity, we will refer to it in this context as $\text{LLL}(\mcal{A},\alpha,S;G)$. The induced network \emph{efficiency} is defined as
\begin{equation}\label{eq:JA}
	\begin{aligned}
		\JB(\alpha,S;G) &:=  \frac{\min_{a\in\text{LLL}(\mcal{A},\alpha,S;G)} W(a)}{\max_{a'\in\mcal{A}} W(a')} \\ 
		&=  \frac{\min_{a\in\text{LLL}(\mcal{A},\alpha,S;G)} W(a)}{(1+\alphasys)|\mcal{E}|},
	\end{aligned}
\end{equation}
which is the ratio of the welfare induced by the welfare-minimizing SSS  to the optimal welfare. The second equality above is due to the fact that optimal welfare is attained at $\vec{x}$ (all play $x$). We re-iterate that the imposter nodes serve only to modify the stochastically stable states, and do not contribute to the system welfare $W(a)$ \eqref{eq:welfare}. The \emph{risk} from broad attacks faced by the system operator in choosing gain $\alpha$  is defined as
\begin{equation}\label{eq:imposter_risk}
	\RB(\alpha,S;G) := 1 - \JB(\alpha,S;G)  .
\end{equation}
Risk measures the distance from optimal efficiency under operating gain $\alpha$. Fig. \ref{fig:GA_example} illustrates an example of a three-node network subject to a {broad} adversarial attack. The extent to which systems are susceptible to {broad} attacks is captured by the following definition of worst-case risk.
\begin{definition}
	The \emph{worst-case risk} to  {broad} attacks is given by
    	\begin{equation}\label{eq:RAstar}
    		R_{\text{\emph{b}}}^*(\alpha) := \max_{N\geq 3} \max_{G\in\mcal{G}_N}\max_{S\in\mcal{T}(G)} R_{\text{\emph{b}}}(\alpha,S;G)  , \\
    	\end{equation}
\end{definition}
The quantity $\RBstar(\alpha)$ is the cost metric that the system operator wishes to reduce given uncertainty of the network structure and target set.
\begin{theorem}\label{thm_RA_WC}
    Let $\alpha > 0$. The worst-case broad risk is
    \begin{equation}\label{eq:RA_WC}
        R_{\text{\emph{b}}}^*(\alpha) = 
        \begin{cases}
            1 - \left(\frac{k}{k+1}\right) \left(\frac{1}{1+\alpha_{\text{sys}}}\right) \ &\text{if } \alpha \in  I_k ,\ \text{for } k=1,2,\ldots \\
            1-\frac{1}{1+\alpha_{\text{sys}}} &\text{if } \alpha \in \left[1,\frac{3}{2}\right] \\
            0 &\text{if } \alpha > \frac{3}{2}
        \end{cases}
    \end{equation}
    where
    \begin{equation}\label{eq:Ik}
    	I_k := \left(\frac{k-1}{k},\frac{k}{k+1}\right].
    \end{equation}
\end{theorem}
It is a piecewise constant function on half-open intervals that is monotonically decreasing in $\alpha$. An illustration is given in Figure \ref{fig:RAstar}, along with the graphs and target sets that achieve the worst-case risks. For sufficiently high gains $\alpha > 3/2$, the system is safeguarded from any broad adversarial attack, i.e. the worst-case risk is zero. By inflating the value of the $x$-convention, the adversary is unable to induce any mis-coordinating links or agents to play $y$. The technical results needed for the proof are given in Section \ref{sec:analysis}.

\subsection{{Focused} attacks and worst-case risk metric}\label{sec:environment}
An adversary is able to choose a strict subset of agents and force them to commit to prescribed choices.  This causes them to act as \emph{fixed agents}, or agents that do not update their choices over time. One could consider this  as allowing the adversary an unlimited number of imposter nodes (instead of one) at its dispatch to attach to  each agent in the subset, thereby solidifying their choices. This focused influence on a single agent is stronger than the influence a broad attack has on a single agent in the sense that the latter type does not require the agent to commit to a choice - it merely incentivizes the agent towards one particular choice.  

Let $F_x \subset \mcal{N}$ $(F_y)$ be the set of fixed $x$ ($y$) agents. We call the \emph{fixed set} $F = (F_x,F_y)$, which satisfies $F_x \cap F_y = \varnothing$ and $F_x \cup F_y \subset \mcal{N}$. We call $\mcal{F}(G)$ the set of all feasible fixed sets on a graph $G \in \mcal{G}_N$. A fixed set $F\in\mcal{F}(G)$ restricts the action space to $\mcal{A}(F)$, where $\mcal{A}_i(F) = \{x\}$ ($\{y\}$) $\forall i \in F_x$ ($F_y$) and $\mcal{A}_i(F) = \{x,y\}$ $\forall i \notin F$. We assume the adversary selects at least one fixed agent. The strict subset assumption avoids pathological cases (e.g. alternating $x$ and $y$ fixed nodes for an entire line network yields an efficiency of zero).

The  set of stochastically stable states given a fixed set $F$ is written as $\text{LLL}(\mcal{A}(F),\{U_i^\alpha\}_{i\in\mcal{N}};G)$. However for brevity, we will refer to it as $\text{LLL}(\mcal{A}(F),\alpha;G)$. The induced efficiency is
\begin{equation}\label{eq:JE}
	\JF(\alpha,F;G) :=  \frac{\min_{a\in\text{LLL}(\mcal{A}(F),\alpha;G)} W(a)}{\max_{a \in \mcal{A}(F)} W(a)},
\end{equation}
which is the ratio of the welfare induced by the worst-case stable state to the optimal welfare given the fixed set $F$. The risk faced by the system operator in choosing $\alpha$  is defined as
\begin{equation}
	\RF(\alpha,F;G) := 1- \JF(\alpha,F;G) \ .
\end{equation}
Again, risk measures the distance from optimal efficiency when choosing $\alpha$. The fixed nodes here differ from the imposter nodes in that they contribute to the true measured welfare \eqref{eq:welfare} in addition to modifying the SSS by restricting the action set and influencing the decisions of their non-fixed neighbors.  Figure \ref{fig:GE_example} provides an illustrative  example of a network with three fixed agents and one unfixed agent. The extent to which the system is susceptible to {focused} attacks is defined by the following worst-case risk metric.
\begin{definition}
	The \emph{worst-case risk from {focused} attacks} is given by
	\begin{equation}\label{eq:REstar}
		R_{\text{\emph{f}}}^*(\alpha) := \max_{N\geq 3} \max_{G\in\mcal{G}_N} \max_{F\in\mcal{F}(G) } R_{\text{\emph{f}}}(\alpha,F;G) \ . 
	\end{equation}
\end{definition}
The quantity $\RFstar(\alpha)$ is the cost metric that a system operator wishes to reduce given uncertainty on the graph structure and  composition of fixed agents in the network.  
\begin{theorem}\label{thm_RE_WC}
    The worst-case  risk from {focused} attacks is 
    \begin{equation}\label{eq:RE_WC}
        R_{\text{\emph{f}}}^*(\alpha) = 
        \begin{cases}
            1-\frac{1+\alpha}{1 + \alpha_{\text{sys}}}, \ &\text{if } \alpha < \alpha_{\text{sys}} \\
            0, &\text{if } \alpha = \alpha_{\text{sys}} \\
            1-\frac{1+\alpha_{\text{sys}}}{1 + \alpha},  &\text{if } \alpha > \alpha_{\text{sys}}
        \end{cases}.
    \end{equation}
\end{theorem}
The technical results needed for the proof are given in Section \ref{sec:analysis}. An illustration of this quantity as well as the graphs that induce worst-case risk are portrayed in Figure \ref{fig:REstar}. We observe the choice $\alpha = \alpha_{\text{sys}}$ recovers optimal efficiency for any $G\in \mcal{G}_N$ and $F\in\mcal{F}(G)$. In other words, by operating at the system gain $\alpha_{\text{sys}}$, the system operator safeguards efficiency from any {focused} attack.  Furthermore, $\RFstar(\alpha)$ monotonically increases for $\alpha > \alpha_{\text{sys}}$, approaching 1 in the limit $\alpha \rightarrow \infty$. Intuitively, the risk in this regime comes from inflating the benefit of the $x$ convention, which can be harmful to system efficiency when there are predominantly  fixed $y$ nodes in the network.  For $\alpha < \alpha_{\text{sys}}$, $\RFstar(\alpha)$ monotonically decreases. The risk here stems from de-valuing the $x$ convention, which hurts efficiency when coordinating with  fixed $x$ nodes is more valuable than coordinating with fixed $y$ nodes.

\subsection{Fundamental tradeoffs between risk and security}
We describe the operator's tradeoffs between the two worst-case risk metrics. That is, given a level of security $\gamma \in [0,1]$ is ensured on one worst-case risk, what is the minimum achievable risk level of the other? These relations are direct consequences of Theorems \ref{thm_RA_WC} and \ref{thm_RE_WC}.
\begin{remark}
	Before presenting the tradeoff relations, we first observe that since $R_{\text{\emph{f}}}^*(\alpha)$ is decreasing on $\alpha < \alphasys$ and  $R_{\text{\emph{b}}}^*(\alpha)$ is decreasing in $\alpha$, the operator should not select any gain $\alpha < \alphasys$, as it worsens both risk levels. Hence for the rest of this paper, we only consider gains greater than $\alphasys$.
\end{remark}

\begin{corollary}\label{thm_gammaE}
	Fix $\gamma_{\text{\emph{f}}} \in [0,1)$. Suppose $R_{\text{\emph{f}}}^*(\alpha) \leq \gamma_{\text{\emph{f}}}$ for some $\alpha$. Then 
	\begin{equation}
		R_{\text{\emph{b}}}^*(\alpha) \geq R_{\text{\emph{b}}}^*\left( \frac{1+\alpha_{\text{sys}}}{1-\gamma_{\text{\emph{f}}} }-1 \right).
	\end{equation}
\end{corollary}
\begin{proof}
	From \eqref{eq:RE_WC}, $\RFstar(\alpha) \leq \gammaF$ implies $\alpha \leq \frac{1+\alpha_{\text{sys}}}{1-\gammaF} - 1$. Since $\RBstar(\alpha)$ is a decreasing function in $\alpha$, we obtain the result.
\end{proof}
In words, as the security from worst-case {focused} attacks improves ($\gammaF$ lowered), the risk from worst-case {broad} attacks increases. A tradeoff relation also holds in the opposite direction.

\begin{corollary}\label{thm_gammaA}
	Fix $\gamma_{\text{\emph{b}}} \in \left[0,1 \right]$. Suppose $R_{\text{\emph{b}}}^*(\alpha) \leq \gamma_{\text{\emph{b}}}$ for some $\alpha$. Suppose $\alphasys \in I_{k_{\text{sys}}}$ for some $k_{\text{sys}} \in \{1,2,\ldots\}$. Then
	\begin{equation}
		R_{\text{\emph{f}}}^*(\alpha) 
		\begin{cases} 
			\geq 0 &\text{if } \gamma_{\text{\emph{b}}} \in \left[1 - \frac{k_{\text{sys}}}{k_{\text{sys}}+1} \frac{1}{1+\alphasys},1\right] \\
			> R_{\text{\emph{f}}}^*\left(\frac{k}{k+1} \right) &\text{if } \gamma_{\text{\emph{b}}} \in \left[1-\frac{k}{k+1} \frac{1}{1+\alphasys}, 1 - \frac{k-1}{k} \frac{1}{1+\alphasys} \right) \\
			&\quad \text{for } k = k_{\text{sys}}, k_{\text{sys}}+1, \ldots \\
			\geq R_{\text{\emph{f}}}^*(1) &\text{if } \gamma_{\text{\emph{b}}} = 1 - \frac{1}{1+\alphasys} \\
			> R_{\text{\emph{f}}}^*\left( \frac{3}{2} \right) &\text{if } \gamma_{\text{\emph{b}}} \in \left[0, 1 - \frac{1}{1+\alphasys} \right)
		\end{cases}
	\end{equation}
	If $\alphasys \in [1,3/2]$,
	\begin{equation}
		R_{\text{\emph{f}}}^*(\alpha)  
		\begin{cases} 
			\geq 0 &\text{if } \gamma_{\text{\emph{b}}} \in \left[1 -  \frac{1}{1+\alphasys},1\right] \\
			> R_{\text{\emph{f}}}^*\left( \frac{3}{2} \right) &\text{if } \gamma_{\text{\emph{b}}} \in \left[0, 1 - \frac{1}{1+\alphasys} \right)
		\end{cases}
	\end{equation}
	If $\alpha_{\text{sys}} > \frac{3}{2}$, then $R_{\text{\emph{f}}}^*(\alpha) \geq 0$ for any $\gamma_{\text{\emph{b}}}$.
\end{corollary}
\begin{proof}
	All bounds are computed by finding $\inf_{\alpha} \RFstar(\alpha)$ s.t. $\RBstar(\alpha) \leq \gammaB$. The relations $\geq$ and $>$ follow  from the fact that $\RFstar(\alpha)$ is increasing in $\alpha>\alphasys$, and depending on whether $\RFstar$ can attain the resulting value.
\end{proof}
Here, as the security from worst-case {broad} attacks improves ($\gammaB$ lowered), the risk from worst-case {focused} attacks increases.  Each of the {broad} risk levels can be attained for a range of {focused} risks. An illustration of the attainable worst-case risk levels is given in Fig. \ref{fig:tradeoff_mix} (blue).

\section{Randomized operator strategies}\label{sec:mixing}
In this section, we consider the scenario where the operator randomizes over multiple gains.  We present a definition and a characterization of worst-case expected risks. We then identify the  risk-security tradeoffs available in the randomized gain setting. We observe they significantly improve upon the deterministic gain setting (Fig. \ref{fig:tradeoff_mix}). We then identify ways to further improve these tradeoffs through different randomizations. 

\subsection{Worst-case expected risks}

Suppose the operator selects a gain from the $M$ distinct values $\bm{\alpha} = \{ \alpha_k \}_{k=1}^M$ satisfying $\alpha_1 < \alpha_2 < \cdots < \alpha_M$ with the probability distribution $\bm{p} = [p_1,\ldots,p_M]^\top \in \Delta_M$. Here we denote $\Delta_M = \{ \bm{p} \in \mbb{R}^M_+ : \sum_{j=1}^M p_j = 1 \}$ as the set of all $M$-dimensional probability vectors. In other words, the operator employs the payoff gain $\alpha_j$ with probability $p_j$. 

We consider the following natural definitions of expected risks. Given a graph $G\in\mcal{G}_N$ and target set $S \in \mcal{T}(G)$,  let $\mbb{E}_{\bm{\alpha},\bm{p}} [\RB | S,G] := \sum_{j=1}^M p_j \RB(\alpha_j,S;G)$ be the expected adversarial risk of the operator's strategy $\bm{\alpha},\bm{p}$. The worst-case expected risk from {broad} attacks is defined as 
\begin{equation}
	\mbb{E}_{\bm{\alpha},\bm{p}}^*[\RB] := \max_{N\geq 3} \max_{G\in\mcal{G}_N} \max_{S\in\mcal{T}(G) } \mbb{E}_{\bm{\alpha},\bm{p}}[\RB | S,G] .
\end{equation}
Similarly, given a fixed set $F\in\mcal{F}(G)$, let $\mbb{E}_{\bm{\alpha},\bm{p}}[\RF | F,G] := \sum_{j=1}^M p_i \RF(\alpha_j,F;G)$ be the expected  risk from {focused} attacks. The worst-case expected risk from {focused} attacks is defined as
\begin{equation}
	\mbb{E}_{\bm{\alpha},\bm{p}}^*[\RF] := \max_{N\geq 3} \max_{G\in\mcal{G}_N} \max_{F\in\mcal{F}(G) } \mbb{E}_{\bm{\alpha},\bm{p}}[\RF | F,G] .
\end{equation}

\begin{theorem}\label{mixing_bound}
	Suppose the operator randomizes with gains $\bm{\alpha} =  \{ \alpha_k \}_{k=1}^M$ according to $\bm{p} \in \Delta_M$. Then the worst-case expected {broad} risk is 
	\begin{equation}\label{eq:WC-EA}
			\mbb{E}_{\bm{\alpha},\bm{p}}^*[R_{\text{\emph{b}}}] = \!\max_{k=1,\ldots,M} \left\{ \!\! \left(\sum_{j=1}^k p_j \right) R_{\text{\emph{b}}}^*(\alpha_k)  \right\} .
	\end{equation}
	The worst-case expected {focused} risk is
	\begin{equation}\label{eq:WC-EE}
		\mbb{E}_{\bm{\alpha},\bm{p}}^*[R_{\text{\emph{f}}}] = \!\max_{k=1,\ldots,M} \left\{ \! \!\left(\sum_{j=k}^{M} p_j\!\right)  R_{\text{\emph{f}}}^*(\alpha_k) \right\} .
	\end{equation}
\end{theorem}
The proofs  are given in Section \ref{sec:analysis2}. The characterization of worst-case expected risk is a discounted weighting of a deterministic worst-case risk level. This suggests that the risk levels achievable by randomization can improve upon the risks induced from a deterministic gain.

\subsection{Risk tradeoffs under randomized operator strategies} 

Given a level of security $\gamma \in [0,1]$ is ensured on one expected worst-case metric, what is the the minimum achievable risk level on the other? We find this can be calculated through a linear program. We formalize these tradeoffs in the following two statements, which are analogous to Corollaries \ref{thm_gammaE} and \ref{thm_gammaA}. 


\begin{corollary}\label{thm_bmix}
	Fix $\gamma_{\text{\emph{f}}} \in [R_{\text{\emph{f}}}^*(\alpha_1),1]$ and a set of gains $\bm{\alpha} = \{\alpha_j\}_{j=1}^M$. Suppose $\mbb{E}_{\bm{\alpha},\bm{p}}^*[R_{\text{\emph{f}}}] \leq \gamma_{\text{\emph{f}}}$ for some $\bm{p} \in \Delta_M$. Then 
	\begin{equation}
		\mbb{E}_{\bm{\alpha},\bm{p}}^*[R_{\text{\emph{b}}}] \geq v_{\text{\emph{b}}}(\gamma_{\text{\emph{f}}},\bm{\alpha}),
	\end{equation}
	where $v_{\text{b}}(\gamma,\bm{\alpha})$ is the value of the following linear program.
	\begin{equation}\label{eq:LP1}
		\begin{aligned}
			v_{\text{b}}(\gamma_{\text{\emph{f}}},\bm{\alpha}) &= \min_{\bm{p}', v} v \\
			&\text{\emph{s.t.} } \sum_{i=1}^M p'_i = 1, \ p_i \geq 0 \ \forall i = 1,\ldots,M \\
			&v \in [0,1] \\
			&A_{\text{\emph{LP}}} \begin{bmatrix} \bm{p}' \\ v \end{bmatrix} \preceq \begin{bmatrix} 0_M \\ \gamma_{\text{\emph{f}}} \mathds{1}_M \end{bmatrix} \\
		\end{aligned}
	\end{equation}
	where $\preceq$ denotes elementwise $\leq$, $0_M$ and $\mathds{1}_M$ are column $M$-vectors of zeros and ones respectively, and $A_{\text{\emph{LP}}}$ is the $2M \times (M+1)$ matrix
	\begin{equation}\label{eq:LPmatrix1}
		{\small A_{\text{\emph{LP}}} = \left[
			\begin{array}{cccc;{2pt/2pt}c}
				R_{\text{\emph{b}}}^*(\alpha_1) & 0 & \cdots & 0  & -1 \\
				R_{\text{\emph{b}}}^*(\alpha_2) & R_{\text{\emph{b}}}^*(\alpha_2) & \cdots & \vdots &  \vdots \\
				\vdots & & \ddots & 0 &   \\
				R_{\text{\emph{b}}}^*(\alpha_M) & \cdots & \cdots & R_{\text{\emph{b}}}^*(\alpha_M) &  -1 \\ \hdashline[2pt/2pt]
				R_{\text{\emph{f}}}^*(\alpha_1) & \cdots & \cdots & R_{\text{\emph{f}}}^*(\alpha_1) &  0 \\
				0 & R_{\text{\emph{f}}}^*(\alpha_2) & \cdots & R_{\text{\emph{f}}}^*(\alpha_2) &  \vdots  \\
				\vdots & & \ddots & \vdots &   \\
				0 & \cdots & 0 & R_{\text{\emph{f}}}^*(\alpha_M) & 0 
			\end{array} \right]}.
	\end{equation}
	Moreover, $v_{\text{b}}(\gamma_{\text{\emph{f}}},\bm{\alpha})$ is decreasing in $\gamma_{\text{\emph{f}}}$.
\end{corollary}
\begin{proof}
	We need to show equivalence between the linear program \eqref{eq:LP1} and the optimization problem
	\begin{equation}\label{eq:optimization}
		\min_{\bm{p}' \in \Delta_M} \mbb{E}^*_{\bm{\alpha},\bm{p}'}[\RB] \ \text{subject to} \ \mbb{E}^*_{\bm{\alpha},\bm{p}'}[\RF] \leq \gamma_{\text{f}}.
	\end{equation}
	Let $A_{\text{b}}(\bm{\alpha}) \in \mbb{R}^{M\times M}$ be the matrix defined by the upper left block of \eqref{eq:LPmatrix1} and $A_{\text{f}}(\bm{\alpha})$ by the bottom left block. From Theorem \ref{mixing_bound}, we can express $\mbb{E}_{\bm{\alpha},\bm{p}'}^*[R_{\text{b}}]$ as the maximum element of the $M$-vector $A_{\text{b}}(\bm{\alpha})\bm{p}'$, and similarly $\mbb{E}_{\bm{\alpha},\bm{p}'}^*[R_{\text{f}}]$ as the maximum element of $A_{\text{f}}(\bm{\alpha})\bm{p}'$. Hence, $\mbb{E}^*_{\bm{\alpha},\bm{p}'}[\RF] \leq \gamma_{\text{f}}$ is the linear constraint $[A_{\text{f}}(\bm{\alpha})\bm{p}']_i \leq \gamma$ for all $i=1,\ldots,M$. The objective $\min_{\bm{p}' \in \Delta_M} \mbb{E}^*_{\bm{\alpha},\bm{p}'}[\RB]$ itself can be cast as a linear objective with linear constraints, i.e. $\min_{\bm{p}' \in \Delta_M,v \in [0,1]} v \ \text{s.t.} \ [A_{\text{b}}(\bm{\alpha})\bm{p}']_i \leq v$. Combining these two, we obtain \eqref{eq:LP1}. The claim $v_{\text{b}}(\gamma,\bm{\alpha})$ is decreasing in $\gamma$ follows as a consequence of the linear program \eqref{eq:LP1}.
\end{proof}
We note that a worst-case expected focused  risk $\mbb{E}^*_{\bm{\alpha},\bm{p}}[\RF] < \RFstar(\alpha_1)$ is not attainable because $\alpha_1$ is the smallest gain it mixes with. Hence, the linear program \eqref{eq:LP1} is infeasible for $\gammaF < \RF^*(\alpha_1)$.  The following tradeoff relation holds in the opposite direction.

\begin{figure}
	\centering
	\includegraphics[scale=.45]{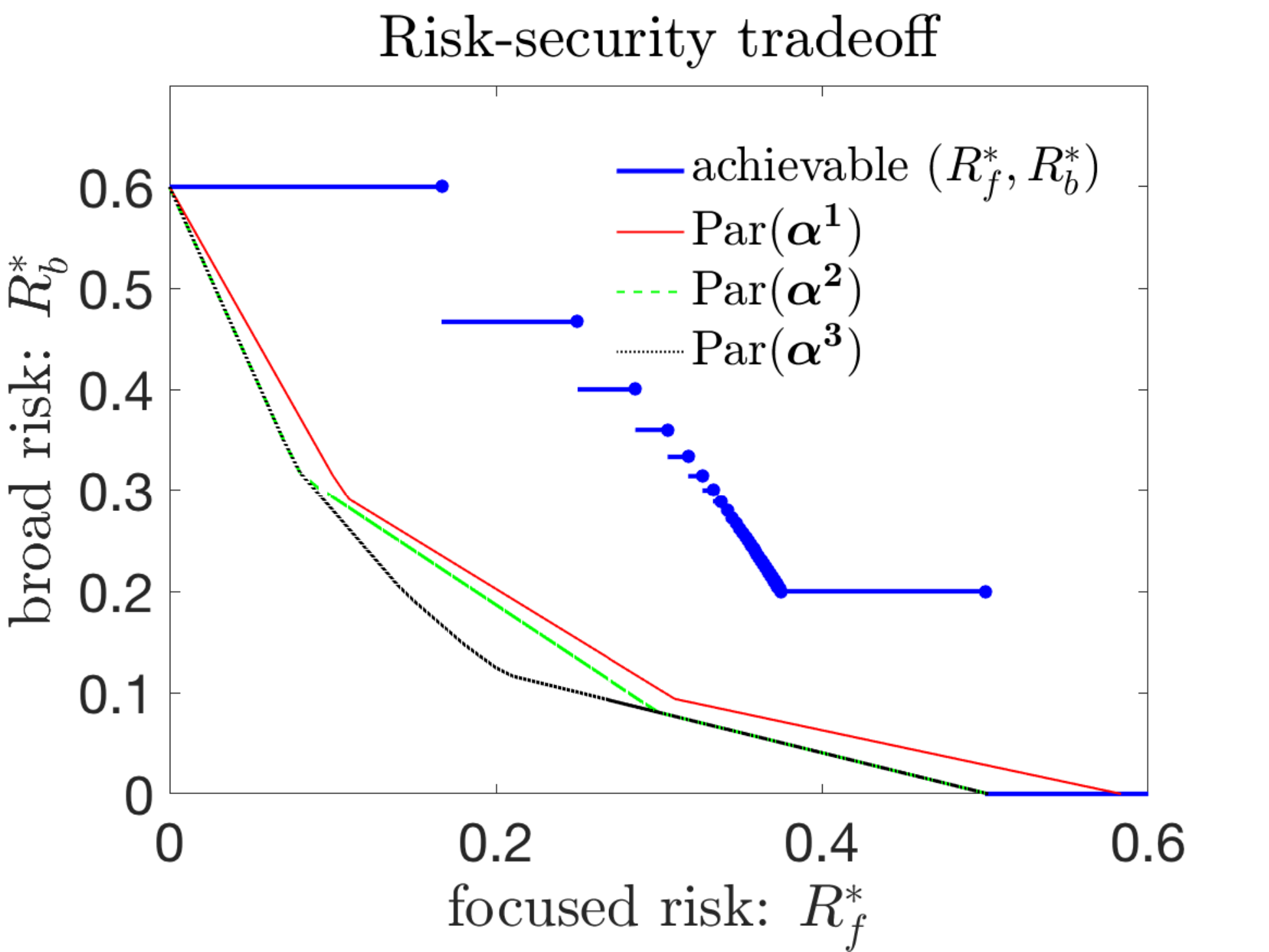}
	\caption{\small  Security-risk tradeoffs are depicted by the achievable worst-case risk levels from deterministic gains (blue) and randomized gains (red, green, black). The Pareto frontiers for three different randomized strategies $\bm{\alpha}^1,\bm{\alpha}^2 \in \mbb{R}_+^5$, and $\bm{\alpha}^3 \in \mbb{R}_+^{300}$, are shown in increasing order of improvement.  The strategies $\bm{\alpha}^1$ and $\bm{\alpha}^2$ randomize over the highest three {broad} risk levels in addition to the lowest two. The strategy $\bm{\alpha}^3$ randomizes over the highest 298 {broad} risk levels and the lowest two. We chose the values as follows. For $k = 1,2$, we set $\alpha_1^k = \alphasys$, $\alpha_j^k  = (1-\epsilon_k) \frac{j-1}{j}  + \epsilon_k \frac{j}{j+1}\in I_j$ for $j = 2,3$, $\alpha_4^k = 1+\epsilon_k$, and $\alpha_5^k = \frac{3}{2} + \epsilon_k$. We have set $\epsilon_1 = 0.5$ and $\epsilon_2 = .01$. Hence, Par($\bm{\alpha}^2$) improves upon Par($\bm{\alpha}^1$) via Claim \ref{pareto_incr}. For $k=3$,  we set $\alpha_1^3 = \alphasys$, $\alpha_j^3  = (1-\epsilon_3) \frac{j-1}{j}  + \epsilon_3 \frac{j}{j+1}\in I_j$, $j = 2,3,\ldots,298$, $\alpha_{299}^3 = 1+\epsilon_3$, and $\alpha_5^3 = \frac{3}{2} + \epsilon_3$. Claim \ref{pareto_more} ensures Par($\bm{\alpha}^3$) improves upon Par($\bm{\alpha}^2$). We chose $\epsilon_3 = .01$ and  $\alpha_{\text{sys}} = 1/4$.} 
	\label{fig:tradeoff_mix}
\end{figure}

\begin{corollary}\label{thm_fmix}
	Fix $\gamma_{\text{\emph{b}}} \in [R_{\text{\emph{b}}}^*(\alpha_M),1]$ and a set of gains $\bm{\alpha} = \{\alpha_j\}_{j=1}^M$. Suppose $\mbb{E}^*_{\bm{\alpha},\bm{p}}[R_{\text{\emph{b}}}]  \leq \gamma_{\text{\emph{b}}}$ for some $\bm{p} \in \Delta_M$. Then 
	\begin{equation}
		\mbb{E}^*_{\bm{\alpha},\bm{p}}[R_{\text{\emph{f}}}] \geq v_{\text{\emph{f}}}(\gamma_{\text{\emph{b}}},\bm{\alpha}), 
	\end{equation}
	where $v_{\text{\emph{f}}}(\gamma_{\text{\emph{b}}},\bm{\alpha})$ is the value of the following linear program.
	\begin{equation}\label{eq:LP2}
		\begin{aligned}
			v_{\text{\emph{f}}}(\gamma_{\text{\emph{b}}},\bm{\alpha}) &= \min_{\bm{p}, v} v \\
			&\text{\emph{s.t.} } \sum_{i=1}^M p_i = 1, \ p_i \geq 0 \ \forall i = 1,\ldots,M \\
			&v \in [0,1] \\
			&\left[\begin{array}{c;{2pt/2pt}c}  
				A_{\text{\emph{f}}}(\bm{\alpha}) & -\mathds{1}_M \\ \hdashline[2pt/2pt]
				A_{\text{\emph{b}}}(\bm{\alpha}) & 0_M
			\end{array}\right] \begin{bmatrix} \bm{p} \\ v \end{bmatrix} \preceq \begin{bmatrix} 0_M \\ \gamma_{\text{\emph{b}}} \mathds{1}_M \end{bmatrix},
		\end{aligned}
	\end{equation}
	where $A_{\text{\emph{f}}}(\bm{\alpha})$ and $A_{\text{\emph{b}}}(\bm{\alpha})$ are defined as the bottom and top left blocks of \eqref{eq:LPmatrix1}, respectively. Furthermore, $v_{\text{\emph{f}}}(\gamma_{\text{\emph{b}}},\bm{\alpha})$ is decreasing in $\gamma_{\text{\emph{b}}}$.
	
\end{corollary}
We omit the proof as it is similar to that of Corollary \ref{thm_bmix}. Note a worst-case expected broad risk $\mbb{E}^*_{\bm{\alpha},\bm{p}}[\RB] < \RBstar(\alpha_M)$ is not attainable since $\alpha_M$ is the highest gain it mixes with - \eqref{eq:LP2} is infeasible for $\gammaB < \RBstar(\alpha_M)$.   Fig. \ref{fig:tradeoff_mix} plots the best achievable risk levels of three randomized operator strategies (red, green, and black). 

\subsection{Improvement of risk tradeoffs}

The tradeoff relations describe the best achievable level on one risk metric given the other is subject to a security constraint when the gains $\bm{\alpha}$ are fixed. One way to improve the achievable risks is to  decrease the available gains.  
\begin{claim}\label{pareto_incr}
	Let $\bm{\alpha},\bm{\alpha}' \in \mbb{R}^M$. Suppose $\alpha_j \in  I_{k_j}$ (recall \eqref{eq:Ik}), $j=1,\ldots,M$  for some non-decreasing subsequence $k_j \geq 1$. Let $\bm{\alpha}'$ satisfy $\alpha_j' \in I_{k_j}$ with $\alpha_j' < \alpha_j$.  Then for all $\gamma_{\text{\emph{b}}} \in [R_{\text{\emph{b}}}^*(\alpha_M),1]$, $v_{\text{\emph{f}}}(\gamma_{\text{\emph{b}}},\bm{\alpha}') \leq v_{\text{\emph{f}}}(\gamma_{\text{\emph{b}}},\bm{\alpha})$. Similarly, for all $\gamma_{\text{\emph{f}}} \in [R_{\text{\emph{f}}}^*(\alpha_1),1]$, $v_{\text{\emph{b}}}(\gamma_{\text{\emph{f}}},\bm{\alpha}') \leq v_{\text{\emph{b}}}(\gamma_{\text{\emph{f}}},\bm{\alpha})$.
\end{claim}

Randomizing over additional gains can also improve the achievable risks.
\begin{claim}\label{pareto_more}
	Suppose $\bm{\alpha} \in \mbb{R}^M$ and $\bm{\alpha}' \in \mbb{R}^{M'}$ with $M < M'$, and assume $\bm{\alpha}'$ contains the elements of $\bm{\alpha}$. Then the assertion of Claim \ref{pareto_incr}  holds. 
\end{claim}
The proofs of the above two Claims follow directly from the formulation of the LPs \eqref{eq:LP1}, \eqref{eq:LP2}, and hence we omit them.

Fig. \ref{fig:tradeoff_mix} depicts the best achievable risk levels of three randomized operator strategies of increasing improvement due to Claims \ref{pareto_incr} and \ref{pareto_more} (red, green, and black curves). In particular, these plots constitute the \emph{Pareto frontier} of all attainable expected risks among distributions $\bm{p}$ given a fixed set of gains. That is, for any $\bm{\alpha}$, we say a risk level $\begin{bmatrix} \mbb{E}_{\bm{\alpha},\bm{p}}[\RF] \\ \mbb{E}_{\bm{\alpha},\bm{p}}[\RB] \end{bmatrix} \in \mbb{R}^2$ belongs to the frontier Par($\bm{\alpha}$) if there does not exist a $\bm{p}' \neq \bm{p}$ such that $\begin{bmatrix} \mbb{E}_{\bm{\alpha},\bm{p}'}[\RF] \\ \mbb{E}_{\bm{\alpha},\bm{p}'}[\RB] \end{bmatrix} \preceq \begin{bmatrix} \mbb{E}_{\bm{\alpha},\bm{p}}[\RF] \\ \mbb{E}_{\bm{\alpha},\bm{p}}[\RB] \end{bmatrix}$.
Within Par($\bm{\alpha}$), the operator can only improve upon one worst-case risk metric by sacrificing performance on the other. 

From Corollary \ref{thm_fmix}, the  frontier given gains $\bm{\alpha}$ is the set of points
\begin{equation}
		\begin{aligned}
			\text{Par}(\bm{\alpha}) = \left\{\begin{bmatrix}   v_{\text{f}}(\gamma_{\text{b}},\bm{\alpha}) \\ \gamma_{\text{b}} \end{bmatrix} \in \mbb{R}^2 :  \gamma_{\text{b}} \in [\RBstar(\alpha_M),\RBstar(\alpha_1)] \right\}.
		\end{aligned}
\end{equation} 
The parameter $\gammaB$ is upper bounded here by $\RBstar(\alpha_1)$ since any risk level with $\mbb{E}_{\bm{\alpha},\bm{p}}[\RB] > \RBstar(\alpha_1)$ is unattainable under $\bm{\alpha}$. Hence, the values $v_{\text{f}}(\gamma_{\text{b}},\bm{\alpha})$ and $v_{\text{f}}(\RBstar(\alpha_1),\bm{\alpha})$ are equivalent for $\gammaB > \RBstar(\alpha_1)$. The frontiers in Fig. \ref{fig:tradeoff_mix} are generated by numerically solving the linear program \eqref{eq:LP2} for a finite grid of points $\gammaB \in [\RBstar(\alpha_M),\RBstar(\alpha_1)]$.

As we have seen, the transition from deterministic to randomized gains ensures a reduction of risk levels. Randomizing over only a few different gains substantially improves upon the attainable deterministic worst-case risks. However, a detailed quantification of such improvements remains a challenge due to the high dimensionality of the model. In particular, we have yet identified a ``limit" frontier that could be obtained by repeated modifications to the gain vector detailed by Claims \ref{pareto_incr} and \ref{pareto_more}.

\section{Proof of Theorems \ref{thm_RA_WC} and \ref{thm_RE_WC}: Deterministic worst-case risks}\label{sec:analysis}
In this section, we develop the technical results that characterize the worst-case risk metrics $\RBstar(\alpha)$  and  $\RFstar(\alpha)$  (Theorems  \ref{thm_RA_WC} and \ref{thm_RE_WC}). Before presenting the proofs, we first present some preliminaries on potential games \cite{Monderer_1996}, which are essential to calculating stochastically stable states. We then define relevant notations for the forthcoming analysis.

\subsection{Potential games}  
Graphical coordination games fall under the class of potential games - games where individual utilities $\{U_i\}_{i\in\mcal{N}}$ are aligned with a global objective, or potential function.  A game is a potential game if there exists a potential function $\phi:\mcal{A} \rightarrow \mbb{R}$ which satisfies
\begin{equation}
	\phi(a_i,a_{-i}) - \phi(a_i',a_{-i}) = U_i(a_i,a_{-i}) - U_i(a_i',a_{-i})
\end{equation} 
for all $i\in\mcal{N}$, $a\in\mcal{A}$, and $a_i' \neq a_i$ \cite{Monderer_1996}. In potential games, the set of stochastically stable states \eqref{eq:LLL} are precisely the action profiles that maximize the potential function \cite{Blume_1995,Marden_2012}. Specifically, $\text{LLL}(\mcal{A},\{U_i\}_{i\in\mcal{N}};G) = \argmax{a\in\mcal{A}} \phi(a)$. Our analysis relies on characterizing a potential function for the graphical coordination game in the presence of adversarial influences. This allows us to compute stochastically stable states in a straightforward manner.   

\subsection{Relevant notations for analysis}
Any  action profile $a$ on a graph $G = (\mcal{N},\mcal{E}) \in \mcal{G}_N$ decomposes $\mcal{N}$ into $x$ and $y$-partitions.  A node that belongs to a $y$-partition ($x$-partition) has $a_i=y$ ($x$).  The partitions are enumerated $\{ \mcal{P}_y^1, \ldots, \mcal{P}_y^{k_y}\}$ and $\{  \mcal{P}_x^1, \ldots, \mcal{P}_x^{k_x} \}$, are mutually disjoint, and cover the graph. Each partition is a connected subgraph of $G$.  It is possible that $k_x=0$ with $k_y=1$ (when $a=\vec{y}$), $k_x=1$ with $k_y = 0$ (when  $a=\vec{x}$), or $k_y,k_x \geq 1$.  

For any subset of nodes $A,B \subseteq \mcal{N}$, let us denote
\begin{equation}
	e(A,B) :=  \{(i,j) \in \mcal{E} : i\in A, j \in B \} 
\end{equation}
as the set of edges between $A$ and $B$. We write $A^c$ as the complement of $A$. We extensively use the notation
\begin{equation}
	W^{\alpha}(E,a) := \sum_{(i,j) \in E} V^\alpha(a_i,a_j) 
\end{equation}
as the welfare due to edge set $E \subseteq \mcal{E}$ in action profile $a$, where $V^\alpha$ is of the form \eqref{eq:base_game} with $\alpha_{\text{sys}}$ replaced by $\alpha$. For compactness, we will denote $W(E,a)$ as $W^{\alphasys}(E,a)$ for the local system welfare generated by the edges $E$. Our analysis will also rely on the following mediant inequality.
\begin{fact}
	Suppose $n_i \geq 0$ and $d_i > 0$ for each $i=1,\ldots,m \in \mbb{N}$. Then 
	\begin{equation}\label{eq:mediant}
		\frac{\sum_{i=1}^m n_i}{\sum_{i=1}^n d_i} \geq \min_{i} \frac{n_i}{d_i}. 
	\end{equation}
	We refer to the LHS above as the \emph{mediant sum} of the $\frac{n_i}{d_i}$.
\end{fact}

\subsection{Characterization of $\RBstar$: worst-case {broad} risk}
To prove Theorem \ref{thm_RA_WC}, we seek a pair $(S,G)$ with $G\in\mcal{G}_N$  of any size $N \geq 3$ and $S\in\mcal{T}(G)$, that minimizes efficiency $\JB(\alpha,S;G)$ (maximizes risk $\RB(\alpha,S;G)$). Our method to find the minimizer is to show any $(S,G)$ can be transformed into a star network with a particular target set that has lower efficiency, when $\alpha < 1$. Thus, in this regime the search for the worst-case graph reduces to the class of star networks of arbitrary size. For $\alpha \geq 1$,  structural properties allow us to deduce the minimal efficiency.

The graphical coordination game defined by $\mcal{A}=\{x,y\}^N$, perceived utilities $\{\tilde{U}_i^\alpha\}_{i\in\mcal{N}}$  \eqref{eq:utilities2}, target set $S$, and graph $G$ falls under the class of potential games \cite{Monderer_1996}. A potential function is given by
\begin{equation}\label{eq:potential2}
	\frac{1}{2}W^{\alpha}(a) + (1+\alpha)\sum_{i\in S_x}\mathds{1}(a_i=x) + \sum_{i\in S_y} \mathds{1}(a_i = y)
\end{equation}
where 
\begin{equation}\label{eq:perceived_welfare}
	W^\alpha(a) := \sum_{i\in\mcal{N}} U_i^\alpha(a).
\end{equation} 
Hence, the stochastically stable states $\text{LLL}(\mcal{A},\alpha,S;G)$ are maximizers of \eqref{eq:potential2}. Suppose  $\hat{a} = \argmin{a \in\text{LLL}(\mcal{A},\alpha,S;G)} W(a)$ is the welfare-minimizing SSS inducing the partitions $\{\mcal{P}_z^k\}_{k=1}^{k_z}$, $z=x,y$.  We can express its efficiency from \eqref{eq:JA} as
\begin{equation}\label{eq:JI_partitions}
	\frac{\sum_{k=1}^{k_y}|e(\mcal{P}_y^k,\mcal{P}_y^k)| + (1+\alpha_{\text{sys}})\sum_{k=1}^{k_x} |e(\mcal{P}_x^k,\mcal{P}_x^k)|}{(1+\alpha_{\text{sys}})(\sum_{k=1}^{k_y} |e(\mcal{P}_y^k,\mcal{N})| + \sum_{k=1}^{k_x} |e(\mcal{P}_x^k,\mcal{P}_x^k)|)}.
\end{equation}
Note the denominator is simply the number of edges in $G$ multiplied by $1\!+\! \alphasys$. From \eqref{eq:potential2},  each $y$-partition $\mcal{P}_y^k$ in $\hat{a}$ satisfies\footnote{Since we are seeking worst-case  pairs $(S,G)$,  we may consider any $y$-partition as only having $y$ imposters placed among its nodes. This is because any $x$ imposters that were placed in a resulting  $y$-partition can be replaced by $y$-imposters and retain stability.  We reflect this generalization in \eqref{eq:CY} and \eqref{eq:CX}, where influence from only $y$ ($x$) imposters is considered. }
\begin{align}\label{eq:CY}
		|\mcal{P}_y^k| + |e(\mcal{P}_y^k,\mcal{P}_y^k)| \geq &\max_{a_{\mcal{P}_y^k}\neq \vec{y}_{\mcal{P}_y^k}} \!\! W^{\alpha}(e(\mcal{P}_y^k,\mcal{N}), (a_{\mcal{P}_y^k},\hat{a}_{-\mcal{P}_y^k})) \nonumber \\ 
		&+ \sum_{i\in\mcal{P}_y^k} \!\! \mathds{1}(a_i = y). \tag{CY}
\end{align}
In words, no subset of agents in $\mcal{P}_y^k$ can deviate from $y$ to improve the collective perceived welfare of $\mcal{P}_y^k$. A similar stability condition holds for each $x$-partition $\mcal{P}_x^k$. 
\begin{align}\label{eq:CX}
		(1\!+\!\alpha)&\!\left(|\mcal{P}_x^k| \!+\! |e(\mcal{P}_x^k,\mcal{P}_x^k)| \right) \geq \nonumber \\ 
		&\max_{a_{\mcal{P}_x^k}\neq \vec{x}_{\mcal{P}_x^k}} \!\!\!W^{\alpha}\!(e(P_x^k,\mcal{N}), (a_{\mcal{P}_x^k},\hat{a}_{-\mcal{P}_x^k})) \nonumber \\ 
		&\quad+ (1+\alpha)\!\!\!\sum_{i\in\mcal{P}_x^k} \!\!\mathds{1}(a_i = x) \tag{CX}
\end{align}
The following result characterizes the threshold on $\alpha$ above which any network is safeguarded from any imposter attack.
\begin{lemma}\label{threshold}
	Let $N\geq 3$. Then $\alpha > \frac{N}{N-1}$ if and only if 
	\begin{equation}
		\min_{G\in\mcal{G}_N} \min_{S \in \mcal{T}(G)} J_{\text{\emph{b}}}(\alpha,S;G) = 1.
	\end{equation}
\end{lemma}
\begin{proof}
	$(\Rightarrow)$ Let $\alpha > \frac{N}{N-1}$. Suppose there is a pair $(S,G)$ with $\JB(\alpha,G,S) < 1$. Then there must exist a $y$-partition $\mcal{P}_y \subset \mcal{N}$. From \eqref{eq:CY},
	\begin{equation}
		|\mcal{P}_y| + |e(\mcal{P}_y,\mcal{P}_y)| \geq (1+\alpha)|e(\mcal{P}_y,\mcal{N})| > 2|e(\mcal{P}_y,\mcal{N})|.
	\end{equation}
	Since $G$ is connected,  $|e(\mcal{P}_y,\mcal{P}_y)| \geq |\mcal{P}_y| - 1$ and there is at least one outgoing link from $\mcal{P}_y$, i.e. $ |e(\mcal{P}_y,\mcal{P}_y^c)| \geq 1$. Consequently,  $|e(\mcal{P}_y,\mcal{N})|  \geq |\mcal{P}_y|$, from which we obtain
	\begin{equation}
		|e(\mcal{P}_y,\mcal{N})| + |e(\mcal{P}_y,\mcal{P}_y)| > 2|e(\mcal{P}_y,\mcal{N})|.
	\end{equation}
	which is impossible.
	
	\noindent $(\Leftarrow)$ Assume $\min_{G\in\mcal{G}_N} \min_{S \in \mcal{T}(G)} \JB (\alpha,S;G) = 1$. Then no $y$-partition can exist for any graph. In particular, \eqref{eq:CY} is violated for $\mcal{P}_y = \mcal{N}$.
	\begin{equation}
		N + |\mcal{E}| < (1+\alpha)|\mcal{E}| \Rightarrow \alpha > \frac{N}{|\mcal{E}|}.
	\end{equation}
	Since $|\mcal{E}| \geq N-1$, we obtain $\alpha > \frac{N}{N-1}$.
\end{proof}
We also deduce the following minimal efficiencies for any graph when $1 \leq \alpha \leq \frac{N}{N-1}$.
\begin{lemma}\label{alpha_one}
	Suppose $N \geq 3$. Then $\alpha \in [1,\frac{N}{N-1}]$ if and only if
	\begin{equation}
		\min_{G\in\mcal{G}_N} \min_{S \in \mcal{T}(G)} J_{\text{\emph{b}}}(\alpha,S;G) = \frac{1}{1+\alpha_{\text{sys}}}.
	\end{equation}
\end{lemma}
\begin{proof}
	The ($\Rightarrow$) direction follows the same argument as Lemma \ref{threshold}. 
	
	($\Leftarrow$) The assumption implies the only $y$-partition that is stabilizable is $\mcal{N}$. Then for any $\mcal{P}_y \subset \mcal{N}$, \eqref{eq:CY} is violated, i.e. 
	\begin{equation}
		|\mcal{P}_y| + |e(\mcal{P}_y,\mcal{P}_y)| < (1+\alpha)|e(\mcal{P}_y,\mcal{N})|.
	\end{equation}
	Since $G$ is connected and there is at least one outgoing edge from $\mcal{P}_y$, we obtain
	 \begin{equation}
		\frac{2|\mcal{P}_y| - 1}{|\mcal{E}|} < 1 + \alpha 
	\end{equation}
	The above holds for any graph $G = (\mcal{N},\mcal{E})$ and subset of nodes $\mcal{P}_y \subset \mcal{N}$. From the facts that $|\mcal{P}_y| \leq N-1$ and $|\mcal{E}| \geq N-1$, we have $\alpha > \frac{N-2}{N-1}$ for any $N \geq 3$. Consequently, $\alpha \geq 1$ and Lemma \ref{threshold} establishes that $\alpha \leq \frac{N}{N-1}$.
\end{proof}


The class of star graphs is central to the worst-case analysis in the interval $0< \alpha < 1$.
\begin{definition}
	Let $\mbb{S}_N$ be the set of all $(S,G)$ where $G$ is the star graph with $N$ nodes, $S_y$ contains the center node, and $S_x = \mcal{N} \backslash S_y$.
\end{definition}
An immediate consequence of this definition is the leaf nodes satisfy \eqref{eq:CX}.  The efficiency is then proportional to the fraction of leaf nodes that are stable to $y$, if any. Furthermore, the stability condition \eqref{eq:CY} of $\mcal{P}_y = S_y$ for members of $\mbb{S}_N$ simplifies to
\begin{equation}\label{eq:star_CY}
	2|e(\mcal{P}_y,\mcal{P}_y)|+1 \geq (1+\alpha)(N-1).
\end{equation}
In other words, stability of the target set $S_y$ as a $y$-partition hinges on \eqref{eq:CY} being satisfied for the selection $a_{\mcal{P}_y} =  \vec{x}$.  The following result reduces the search space for efficiency minimizers to $\mbb{S}_N$ when $\alpha < 1$.
\begin{lemma}\label{star_reduction}
	Suppose $0< \alpha < 1$ and $n\geq 3$. Consider any $(S,G)$ with $G \in \mcal{G}_N$, $S\in\mcal{T}(G)$. Then there is a $(S',G') \in \mbb{S}_{N'}$ such that $J_{\text{\emph{b}}}(\alpha,S';G') \leq J_{\text{\emph{b}}}(\alpha,S;G)$  for some $N' \geq N$.
\end{lemma}
The idea of the proof is to construct a member of $\mbb{S}_{N'}$ by re-casting the $y$ and $x$-partitions of $(S,G)$ as star subgraphs while preserving the same number and type of edges, thus preserving efficiency. Further efficiency reduction can be achieved by converting excess $x$ links into $y$ links in this star configuration. We provide the proof detailing the constructive procedure in the Appendix. We now characterize the minimal efficiency for the star graph of size $N$, $J_N^*(\alpha) := \min_{(G,S)\in \mbb{S}_N} \JB(\alpha,S;G)$ for $\alpha < 1$.
\begin{lemma}\label{Jstar_n}
    Suppose $\alpha < 1$ and fix $N\geq 3$. Then 
	\begin{equation}\label{eq:star_efficiency}
		 J_N^*(\alpha) = \frac{1}{(1+\alpha_{\text{sys}})(N-1)}\left\lceil \frac{(1+\alpha)(N-1) - 1}{2} \right\rceil .
	\end{equation}
\end{lemma}
\begin{proof}
	The goal is to find the smallest $y$-partition of the $n$ star that is still stabilizable under a gain $\alpha$. This is written
	\begin{equation}
		\begin{aligned}
			J^*_N(\alpha) = &\min_{N_y} \frac{1}{1+\alpha_{\text{sys}}}\frac{N_y}{N-1} \\
			\text{s.t.} \ &\begin{cases} 
						N_y \leq N-1 \ &\text{(size of $y$-partition)} \\
						2 N_y + 1 \geq (1+\alpha)(N-1) \ &\text{(stability)}
					\end{cases}
		\end{aligned}
	\end{equation}
	 The smallest integer $N_y$ that satisfies the constraints is $\left\lceil \frac{(1+\alpha)(N-1) - 1}{2} \right\rceil$ for $\alpha \in (0,1)$. 

\end{proof}
\begin{proof}[Proof of Theorem \ref{thm_RA_WC}]
	For $\alpha < 1$, by Lemma \ref{star_reduction}, the worst-case efficiency is
	\begin{equation}
		\min_{N\geq 3} \min_{(G,S) \in \mbb{S}_N} \JB(\alpha,S;G) = \min_{N\geq 3} J_N^*(\alpha).
	\end{equation}
	Using the formula of Lemma \ref{Jstar_n}, we obtain the first entry in \eqref{eq:RA_WC}.  Lemma \ref{alpha_one} asserts the  minimal  efficiency is $\frac{1}{1+\alpha_{\text{sys}}}$ for $\alpha \in [1,\frac{3}{2}]$ because the upper bound $\frac{N}{N-1}$ is maximized at $N=3$ (for $N\geq 3$).  This gives the second entry in \eqref{eq:RA_WC}.  Lastly, Lemma \ref{threshold} asserts the minimal efficiency is 1 for $\alpha > \frac{3}{2}$.
\end{proof}

\subsection{Characterization of $R_{\text{\emph{f}}}^*$: worst-case {focused} risk}
Our approach for the proof of Theorem  \ref{thm_gammaA}  differs from that of $\RBstar$. Instead of reducing the search of worst-case graphs, we simply provide an upper bound on $\RFstar(\alpha,F;G)$ for any $G$ and fixed set $F\in\mcal{F}(G)$, and show one can construct a graph with fixed nodes that achieves it. .

We observe $\frac{1}{2}W^{\alpha}(a) : \mcal{A}(F) \rightarrow \mbb{R}$ serves as a potential function (recall \eqref{eq:perceived_welfare}) for the game with restricted action set $\mcal{A}(F)$ and  utilities $\{U_i^\alpha\}_{i\in\mcal{N}}$.  Hence, the stochastically stable states $\text{LLL}(\mcal{A}(F),\alpha;G)$ are maximizers of $\frac{1}{2}W^{\alpha}(a)$. Suppose $\hat{a} = \argmin{a \in\text{LLL}(\mcal{A}(F),\alpha;G)} W(a)$
decomposes the graph into the $x$ and $y$-partitions $\{\mcal{P}_z^k\}_{k=1}^{k_z}$, $z=x,y$. We express its efficiency  \eqref{eq:JE} as
\begin{equation}\label{eq:JC_partitions}
	\frac{\sum_{k=1}^{k_y} |e(\mcal{P}_y^k,\mcal{P}_y^k)| + (1+\alpha_{\text{sys}}) \sum_{k=1}^{k_x}|e(\mcal{P}_x^k,\mcal{P}_x^k)|}{\sum_{k=1}^{k_y}W^{\alpha_{\text{sys}}}(e(\mcal{P}_y^k,\mcal{N}),a^*) + \sum_{k=1}^{k_x}W^{\alpha_{\text{sys}}}(e(\mcal{P}_x^k,\mcal{P}_x^k),a^*)}
\end{equation}
where $a^* = \argmax{a \in \mcal{A}(F)} W(a)$ is the welfare-maximizing action profile. Similar to \eqref{eq:CY}, each $y$-partition $\mcal{P}_y^k$ formed from $\hat{a}$ satisfies the stability condition
\begin{equation}\label{eq:CY2}
	|e(\mcal{P}_y^k,\mcal{P}_y^k)| \geq \max_{a_{\mcal{P}_y^k}\neq \vec{y}} W^{\alpha}(e(\mcal{P}_y^k,\mcal{N}), (a_{\mcal{P}_y^k },\hat{a}_{-\mcal{P}_y^k })). \tag{CYE}
\end{equation}
To reduce cumbersome notation, it is understood the max is taken over actions of unfixed nodes, $a_{\mcal{P}_y^k \backslash F}$. Likewise, each $x$-partition $\mcal{P}_x^k$ satisfies
\begin{equation}\label{eq:CX2}
	(1+\alpha)|e(\mcal{P}_x^k,\mcal{P}_x^k)| \geq \max_{a_{\mcal{P}_x^k }\neq \vec{x}} W^{\alpha}(e(\mcal{P}_x^k,\mcal{N}), (a_{\mcal{P}_x^k },\hat{a}_{-\mcal{P}_x^k })) \tag{CXE}.
\end{equation}
The following lemma asserts that agents playing $y$ in the SSS under the gain $\alpha$ remain playing $y$ under a lower gain $\alpha' < \alpha$. The result is crucial for establishing a lower bound on efficiency for any graph $G$ with arbitrary fixed set $F \in \mcal{F}(G)$.

\begin{lemma}\label{Py_stable}
	Suppose $\alpha' < \alpha$.  Denote  $\hat{a}' = \argmin{a \in \text{LLL}(\mcal{A}(F),\alpha';G)} W(a)$ as the welfare-minimizing SSS under $\alpha'$. Then for any $y$-partition $\mcal{P}_y$ induced from $\alpha$, $\hat{a}_i' = y$  for all $ i\in \mcal{P}_y \backslash F$.
\end{lemma}
\begin{proof}
	Condition \eqref{eq:CY2} asserts for all $a_{\mcal{P}_y} \neq \vec{y}$ that
	\begin{equation}
		W^{\alpha}(e(\mcal{P}_y,\mcal{N}),(\vec{y}_{\mcal{P}_y}, \hat{a}_{-\mcal{P}_y})) \geq W^{\alpha} (e(\mcal{P}_y,\mcal{N}),(a_{\mcal{P}_y},\hat{a}_{\mcal{P}_y} )).
	\end{equation}
	It also holds for all $a_{\mcal{P}_y} \neq \vec{y}$ and for any $a_{-\mcal{P}_y} \neq \hat{a}_{-\mcal{P}_y}$ that
	\begin{equation}
		W^{\alpha}(e(\mcal{P}_y,\mcal{N}),(\vec{y}_{\mcal{P}_y}, a_{-\mcal{P}_y} )) \geq W^{\alpha} (e(\mcal{P}_y,\mcal{N}),(a_{\mcal{P}_y},a_{-\mcal{P}_y} ))
	\end{equation}
	because any $y$-links garnered in the RHS above by changing $\hat{a}_{-\mcal{P}_y}$ to $a_{-\mcal{P}_y}$ also contribute to the LHS. In particular, the above holds for $a_{-\mcal{P}_y} = \hat{a}_{-\mcal{P}_y}'$. Lowering the gain to $\alpha'$ preserves the above inequality as well, as it de-values $x$-links garnered on the RHS.
\end{proof}

A dual statement holds - agents playing $x$ in the SSS under $\alpha$ remain so under a higher gain $\alpha' > \alpha$.
\begin{lemma}\label{Px_stable}
	Suppose $\alpha' > \alpha$.  Then for any $x$-partition $\mcal{P}_x$ induced from $\alpha$, $\hat{a}_i' = x$  for all $ i\in \mcal{P}_x\backslash F$.
\end{lemma}
We omit the proof for brevity, as it is analogous to the proof of Lemma \ref{Py_stable}. We are now ready to prove Theorem \ref{thm_gammaA}.

\begin{proof}[Proof of Theorem \ref{thm_gammaA}]
	Consider any graph $G\in\mcal{G}_N$ with fixed set $F$. Recall that efficiency is one for $\alpha = \alphasys$. Thus, we first consider $\alpha < \alpha_{\text{sys}}$. Observe that
	\begin{equation}
		\begin{aligned}
		|e(\mcal{P}_y^k,\mcal{P}_y^k)| &\geq W^{\alpha}(e(\mcal{P}_y,\mcal{N}),(a_{\mcal{P}_y}^*,\hat{a}_{-\mcal{P}_y})) \\
		&=  W^{\alpha}(e(\mcal{P}_y,\mcal{N}),(a_{\mcal{P}_y}^*,a_{-\mcal{P}_y}^*))
		\end{aligned}
	\end{equation} 
	where the inequality is due to \eqref{eq:CY2}. The equality results from Lemma \ref{Px_stable} - the agents ($\notin \mcal{P}_y$) that neighbor any member of $\mcal{P}_y$ remain playing $x$ in $a^*$. We then obtain
\begin{equation}
	\frac{|e(\mcal{P}_y^k,\mcal{P}_y^k)|}{W(e(\mcal{P}_y^k,\mcal{N}),a^*)} \geq \frac{1+\alpha}{1+\alpha_{\text{sys}}}.
\end{equation}
The  inequality results since the   expressions of the numerator and denominator garner the same edges for welfare. It occurs with equality if and only if $a_i^*=x$ $\forall i \in \mcal{P}_y^k\backslash F$. Applying the mediant inequality \eqref{eq:mediant} to  \eqref{eq:JC_partitions}, $\JF(\alpha,F;G) \geq  \frac{1+\alpha}{1+\alpha_{\text{sys}}}$. The case when $\alpha > \alpha_{\text{sys}}$ follows analogous arguments. From Lemma \ref{Py_stable}, $|e(\mcal{P}_y^k,\mcal{P}_y^k)|=W^{\alpha}(e(\mcal{P}_y^k,\mcal{P}_y^k),a^*)$. For $x$-partitions,
\begin{equation}
	\begin{aligned}
		 \frac{(1\!+\! \alpha_{\text{sys}})|e(\mcal{P}_x^k,\mcal{P}_x^k)|}{W(e(\mcal{P}_x^k,\mcal{N}),a^*)} &\geq \frac{1\!+\!\alpha_{\text{sys}}}{1\!+\!\alpha}\frac{W^{\alpha}(e(\mcal{P}_x^k,\mcal{N}),a^*)}{W(e(\mcal{P}_x^k,\mcal{N}),a^*) } \\ 
		 &\geq \frac{1\!+\!\alpha_{\text{sys}}}{1\!+\!\alpha}
	\end{aligned}
\end{equation}
where the first inequality is from \eqref{eq:CX2} and the second occurs with equality  if $a_i^*=y$ $\forall i \in \mcal{P}_x^k\backslash F$. From \eqref{eq:mediant} and \eqref{eq:JC_partitions}, $\JF(\alpha,F;G) \geq  \frac{1+\alpha_{\text{sys}}}{1+\alpha}$.
\end{proof}
We have just shown fundamental lower bounds on efficiency for any graph with fixed agents. The bounds are tight as they can be achieved for any gain $\alpha$ by arranging $N_x$ fixed $x$ and $N_y$ fixed $y$ leaf nodes that influence a single unfixed agent in the center of a star graph. If $\alpha < \alpha_{\text{sys}}$, choosing $\frac{N_x}{N_y}= \frac{1}{1+\alpha}$ gives the minimal  efficiency $\frac{1+\alpha}{1+\alpha_{\text{sys}}}$. If $\alpha > \alpha_{\text{sys}}$, choosing $\frac{N_x}{N_y}= \frac{1}{1+\alpha}$ gives the minimal  efficiency $\frac{1+\alpha_{\text{sys}}}{1+\alpha}$.  Note that if $\alpha$ is rational, one could choose finite integers $N_y,N_x$ that achieve such ratios. Recall Figure \ref{fig:REstar} for illustrative examples. However if it is irrational, they must be taken arbitrarily large to better approximate the ratio.

\section{Proof of Theorem \ref{mixing_bound}: worst-case risks under randomized operator designs}\label{sec:analysis2}

Recall a randomized strategy consists of gains $\bm{\alpha} = \{\alpha_i\}_{i=1}^M$ with distribution $\bm{p} \in \Delta_M$. The gains are ordered $\alphasys \leq \alpha_1 < \cdots < \alpha_M$. To prove  Theorem \ref{mixing_bound}, we outline a few technical Lemmas.  The key insight is the expected efficiency of any graph $G$ can be expressed in the form $\sum_{i=1}^M p_i s_i$, where the coefficient $s_i$ is a mediant sum over local efficiencies of partitions in $G$ when gain $\alpha_i$ is used. The following two mathematical facts are the basis of this insight. 
\begin{fact}\label{planes_lemma}
	Let $\nu_i < \frac{n_i}{d_i} \leq 1$ with $r_i \geq 0$ and $n_i,d_i > 0$ for all $i = 1,\ldots,M$ . Then for all $\bm{p}\in\Delta_M$,
	\begin{equation}\label{eq:p_mediant1}
		\sum_{i=1}^M p_i s_i \geq 1 + \min_{i=1,\ldots,M} \left\{  \left( \sum_{j=1}^i p_j \right) (\nu_i -1) \right\}
	\end{equation}
	where $s_i := \frac{\sum_{j=1}^{i-1} d_j + \sum_{j=i}^M n_j}{\sum_{j=1}^M d_j}, \quad i=1,\ldots,M$.
\end{fact}
We provide a proof in the Appendix. The following dual result follows directly.
\begin{fact}\label{planes_reversed}
	For all $\bm{p}\in\Delta_M$,
	\begin{equation}
		\sum_{i=1}^M p_i s_i' \geq 1 + \min_{i=1,\ldots,M} \left\{ \left( \sum_{j=i}^M p_j \right) (\nu_i -1) \right\}
	\end{equation}
	where $s_i' := \frac{\sum_{j=1}^{i} n_j + \sum_{j=i}^M d_j}{\sum_{j=1}^M d_j}, \quad i=1,\ldots,M$.
\end{fact}
\begin{proof}
	The proof follows similarly to Fact \ref{planes_lemma}, where the indices of the $s_i$ coefficients are reversed.
\end{proof}

We will show for any $(S,G)$ that $\mbb{E}_{\bm{\alpha},\bm{p}}[\JB | S,G] = 1 - \mbb{E}_{\bm{\alpha},\bm{p}}[\RB | S,G]$ can be expressed in the form $\sum_{i=1}^M p_i s_i$ from the LHS of \eqref{eq:p_mediant1}. The lower bounds establish worst-case expected efficiencies - and hence risks. The $\nu_i$  correspond to the worst-case deterministic efficiencies $\JBstar(\alpha_i) = 1 - \RBstar(\alpha_i)$ of the $M$ gains and $\frac{n_i}{d_i}$ to local efficiencies of selected partitions in the graph. Fact \ref{planes_lemma} will be used to establish \eqref{eq:WC-EA}, and Fact \ref{planes_reversed} for \eqref{eq:WC-EE} (Theorem \ref{mixing_bound}).  We now identify a structural property required of worst-case graphs.
\begin{lemma}
	A worst-case graph, i.e. a member of $\argmin{G\in\mcal{G}_N, S \in \mcal{T}(G)} \mbb{E}_{\bm{\alpha},\bm{p}}[J_{\text{\emph{b}}} | S,G]$, has no active $x$-links in $\alpha_1$.
\end{lemma}
\begin{proof}
	Any active $x$-links in $\alpha_1$ remain so for all $\{\alpha_i\}_{i=2}^M$. The efficiency corresponding to each gain can be reduced in the following manner. Delete all such $x$-links and associated agents.  For each  mis-coordinating link between an $x$ and $y$ agent that existed, replace with a single link to a newly created isolated agent with an $x$-imposter attached. This preserves the stochastically stable states of all other nodes while reducing efficiency in each gain.
\end{proof}
Intuitively, a graph that has coordinating $x$ nodes in each gain $\alpha_1,\ldots,\alpha_M$ can be modified by removing these links, resulting in a lower efficiency. We are now ready to prove \eqref{eq:WC-EA} (Theorem \ref{mixing_bound}).
\begin{proof}[Proof of \eqref{eq:WC-EA} (Theorem \ref{mixing_bound})]
	Consider any graph $G = (\mcal{N},\mcal{E}) \in \mcal{G}_N$ and $S \in \mcal{T}(G)$. Let us denote the $M$ (worst-case) stochastically stable states that correspond to each gain $\alpha_i$ with $\hat{a}^i$. Define for each $k=1,\ldots,M$
	\begin{equation}
		P^k = \{ n \in \mcal{N} : \hat{a}_n^i = y \ \forall i \leq k, \ \hat{a}_n^i = x \ \forall i > k\}
	\end{equation}
	as the set of nodes that play $y$ in the SSS in $\alpha_1,\ldots,\alpha_k$ and play $x$ in $\alpha_{k+1},\ldots,\alpha_M$. Note that $P^k$ is possibly composed of multiple $y$-partitions. Also note it is possible that $P^k = \varnothing$ for all $k > \bar{m}$ for some $\bar{m} \in \{2,\ldots,M-1\}$ - that is, $\hat{a}^i = \vec{x}$ for all $i > \bar{m}$. We first consider the case when $P^k \neq \varnothing$ for every $k=1,\ldots,M$. 
	
	Let $Q^k := \{ n \in \mcal{N} : \hat{a}_n^k = y\} = \bigcup_{i=k}^M P^i$. Denote $P^x := \{ n \in\mcal{N} : \hat{a}_n^1 = x\} = (Q^1)^c$ as the set of nodes stable to $x$ for all $\alpha_i$. Consider the gain $\alpha_i$ with $i \leq k$. Then the local efficiency $\frac{W(e(P^k,\mcal{N}),\hat{a}^k) }{W(e(P^k,\mcal{N}),(a^*_{P^k},\hat{a}^k_{-P^k})) }$ of $P^k$ is
	\begin{equation}
		\begin{aligned}
			&\frac{|e(P^k,Q^{k+1})| + |e(P^k,P^k)|}{(1\!+\!\alphasys)(|e(P^k\!,P^k)| \!+\! |e(P^k\!,P^x)| \!+\! |e(P^k\!,(Q^k)^c)|)} \\
			&\quad\quad> J_A^*(\alpha_i).
		\end{aligned}
	\end{equation}
	The inequality is due to Proposition \ref{thm_RA_WC}. For gains $\alpha_i$ with $i>k$, the local efficiency of $P^k$ is
	\begin{equation}
		\frac{(1\!+\!\alphasys)(|e(P^k\!,P^k)| \!+\! |e(P^k\!,P^x)| \!+\! |e(P^k\!,(Q^k)^c)|)}{(1\!+\!\alphasys)(|e(P^k\!,P^k)| \!+\! |e(P^k\!,P^x)| \!+\! |e(P^k\!,(Q^k)^c)|)} =1.
	\end{equation}
	Hence, the overall system efficiency under gain $\alpha_i$ is the mediant sum of the local efficiencies of the $P^k$. An application of Fact \ref{planes_lemma} gives the result. The case when $P^k = \varnothing$ for $k>\bar{m} \in \{2,\ldots,M-1\}$ also follows directly from Fact \ref{planes_lemma}. From the notation of  Fact \ref{planes_lemma}, $\frac{n_k}{d_k} = 1$ for $k>\bar{m}$.
\end{proof}

The details for the proof of \eqref{eq:WC-EA} (Theorem \ref{mixing_bound}) follow analogous arguments pertaining to focused attacks. Recall for a graph $G \in \mcal{G}_n$ and restricted action set $\mcal{A}$, we denote $F = F_x \cup F_y \subset \mcal{N}$ as its set of fixed nodes.  Additionally, we restrict attention to gains $\alpha_i \geq \alphasys$, as these are not strictly dominated in the risk curve.  The following structural property holds in a worst-case graph for focused risk. 
\begin{lemma}\label{WC_y}
	A worst-case graph, i.e., a member of $\argmin{G\in\mcal{G}_N,F\in \mcal{F}(G)} \mbb{E}_{\bm{\alpha},\bm{p}}[J_{\text{\emph{f}}} | F,G]$, has no active $y$-links in $\alpha_M$. Additionally, $a_{F^c}^* = \vec{y}$.
\end{lemma}
\begin{proof}
	A graph that has active $y$-links in $\alpha_M$ remain active for all $\alpha_1,\ldots,\alpha_{M-1}$. The efficiency corresponding to each gain can be reduced by removing all such links and keeping the border nodes as fixed $y$ agents. This preserves the stability properties of all other nodes. The claim $a_{F^c}^* = \vec{y}$ follows from Lemma \ref{Py_stable}.
\end{proof}
We are now ready to prove \eqref{eq:WC-EE} in Theorem \ref{mixing_bound}.
\begin{proof}[Proof of \eqref{eq:WC-EE} (Theorem \ref{mixing_bound})]
	Consider any graph $G = (\mcal{N},\mcal{E}) \in \mcal{G}_N$ and fixed nodes $F \in \mcal{F}(G)$. The $M$ stochastically stable states that correspond to each gain $\alpha_i$ are denoted $\hat{a}^i$. Define for each $k=1,\ldots,M$
	\begin{equation}
		P^k = \{ n \in F^c: \hat{a}_n^i = x \ \forall i \geq k, \ \hat{a}_n^i = y \ \forall i < k  \}
	\end{equation}
	as the set of unfixed nodes that play $x$ in the SSS for $\alpha_{k},\ldots,\alpha_M$ and play $y$  in $\alpha_1,\ldots,\alpha_{k-1}$ . Note that it is possible $P^k = \varnothing$ for all $k < \bar{m}$ for some $\bar{m} \in \{2,\ldots,M-1\}$. That is, $a^k_{F^c} = \vec{y}$ for $k=1,\ldots,\bar{m}-1$. We first consider the case when $P^k \neq \varnothing$ for every $k=1,\ldots,M$. 
	
	Let $Q^k = \{n \in F^c : \hat{a}_n^k = y\} = \bigcup_{i=k}^M P^i$.  Consider the gain $\alpha_i$ with $i\geq k$. Then the local efficiency $\frac{W(e(P^k,\mcal{N}), \hat{a}^k) }{W(e(P^k,\mcal{N}),(a^*_{P^k},\hat{a}^k_{-P^k})) }$ of $P^k$ is 
	\begin{equation}
		\begin{aligned}
			&\frac{(1+\alphasys)(|e(P^k,P^k)| \!+ \! |e(P^k,(Q^{k-1})^c)| \!+\! |e(P^k,F_x)|)}{|e(P^k,P^k)| + |e(P^k,Q^k)| + |e(P^k,F_y)|} \\ 
			&\quad\quad> \JFstar(\alpha_i).
		\end{aligned}
	\end{equation}
	Here, we use the convention $|e(P^1,(Q^0)^c)|= 0$. For gains $\alpha_i$ with $i<k$, the local efficiency of $P^k$ is
	\begin{equation}
		\frac{|e(P^k,P^k)| + |e(P^k,Q^k)| + |e(P^k,F_y)|}{|e(P^k,P^k)| + |e(P^k,Q^k)| + |e(P^k,F_y)|} =1.
	\end{equation}
	Hence the overall system efficiency under $\alpha_i$ is the mediant sum of the local efficiencies of the $P^k$. An application of Fact \ref{planes_reversed} gives the result. The case when $P^k = \varnothing$ for $k>\bar{m} \in \{2,\ldots,M-1\}$ also follows directly from Fact \ref{planes_reversed}. 
\end{proof}

\section{Summary}\label{sec:conclusion}
In this paper, we framed graphical coordination games as a distributed system subject to two types of adversarial influences.  The focus of our study concerned the performance of a class of distributed algorithms against the associated worst-case risks. We identified fundamental tradeoffs between ensuring security against one type of risk and vulnerability to the other, and vice versa. Furthermore, our analysis shows randomized algorithmic designs significantly improves the available tradeoffs.  Our work highlights the design challenges a system operator 
faces in maintaining the efficiency of networked, distributed systems.


\appendix

\begin{figure*}[t]
    \centering
    \includegraphics[scale=1.1]{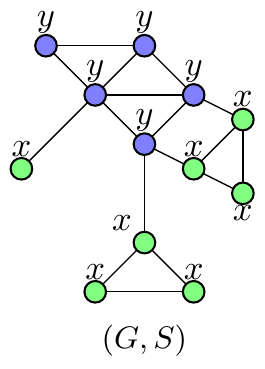}
    \includegraphics[scale=1.1]{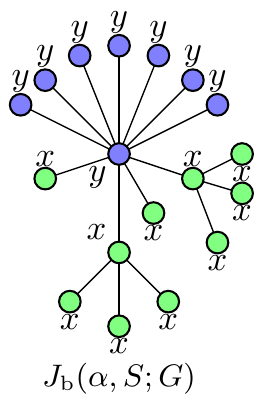}
    \includegraphics[scale=1.1]{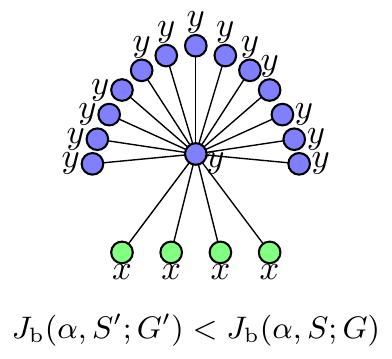}
    \caption{\small An illustration of the constructive process (proof of Lemma \ref{star_reduction}) that generates a member $(S',G')\in\mbb{S}_m$ from any graph $(S,G)$ with one $y$-partition, and $\alpha < 1$. Here, the labels on each node indicate the type of imposter influence. Green (blue) nodes play $x$ ($y$) in the SSS. (Left) Start with an arbitrary graph-adversary pair $(S,G)$. (Center) The partitions of $(S,G)$ are re-cast as star subgraphs with the same number of edges. When there is more than one edge between the $y$ and an $x$-partition, new nodes are created for the excess outgoing edges. This re-casting preserves the original efficiency $\JB(\alpha,S;G)$. (Right) The active $x$-links are converted into $y$-links by redirecting them to the center of the $y$-partition. This results in a graph $(S',G') \in \mbb{S}_m$. }
    \label{fig:proof_diagram}
\end{figure*}


\begin{proof}[Proof of Lemma \ref{star_reduction}:]
	This proof outlines a  procedure to transform any $(S,G)$ into a star graph with lower efficiency if $\alpha < 1$. We split into  two cases - either $(S,G)$ induces a single $y$-partition or more than one. First, assume $(S,G)$ induces a single $y$-partition $\mcal{P}_y$. An illustration of the constructive process is shown in Figure \ref{fig:proof_diagram}. 

Construct a star subgraph $\Gamma_y$ that has $1+|e(\mcal{P}_y,\mcal{P}_y)|$ nodes, each having a $y$ imposter attached. Call the center node $i_y$. Construct similar star configurations $\Gamma_x^k$ for each $x$-partition $\mcal{P}_x^k$. Call their center nodes $i_x^k$. Connect $\Gamma_y$ to each $\Gamma_x^k$ with a  link between $i_x^k$ and $i_y$. If there are multiple edges between $\mcal{P}_x^k$ and $\mcal{P}_y$ ($|e(\mcal{P}_x^k,(\mcal{P}_x^k)^c)| \geq 2$), create $|e(\mcal{P}_x^k,(\mcal{P}_x^k)^c)|  - 1$ new isolated nodes with a single $x$ imposter attached, and connect each to $i_y$ with a single link. At this point, $\Gamma_y$ and $\Gamma_x^k$ are stable $y$ and $x$-partitions, and  the isolated nodes are stable playing $x$. We have obtained a graph of $N' \geq N$ nodes with identical efficiency to $(S,G)$ since the number and type of edges are preserved. 

We can further reduce efficiency if there are active $x$ links, i.e. if $|e(\Gamma_x^k,\Gamma_x^k)| \geq 1$ for at least one $\Gamma_x^k$. If there are none, then the graph belongs to $\mbb{S}_{N'}$ and we are done. Otherwise for each leaf node $j \in \mcal{P}_x^k$, redirect the edge $(j,i_x^k)$ to $(j,i_y)$, and replace $j$'s $x$ imposter with a $y$ imposter. Call $m_x$ the total number of such converted nodes. The resulting graph-target pair $(S',G')$ belongs to $\mbb{S}_{N'}$. We claim the resulting (larger) $y$-partition $\Gamma_y'$ is stable.  For this claim to hold, \eqref{eq:star_CY} requires that
\begin{equation}
	2|e(\mcal{P}_y,\mcal{P}_y)| + 2m_x + 1 \geq (1+\alpha)(|e(\mcal{P}_y,\mcal{N})|+m_x).
\end{equation}
From the original $\mcal{P}_y$, it holds that
\begin{equation}
	\begin{aligned}
		&|e(\mcal{P}_y,\mcal{P}_y)| + |\mcal{P}_y| \geq (1+\alpha)|e(\mcal{P}_y,\mcal{N})| \\ 
		&\Rightarrow 2|e(\mcal{P}_y,\mcal{P}_y)| + 2m_x + 1 > (1+\alpha)(|e(\mcal{P}_y,\mcal{N})|+m_x)
	\end{aligned}
\end{equation}
due to $|\mcal{P}_y| \leq 1 + |e(\mcal{P}_y,\mcal{P}_y)|$ and $\alpha<1$.   All $x$-partitions  in $(S',G')$ , now just a collection of single nodes connected to $i_y$ with an $x$-imposter, are stable. The efficiency is less than the original because active $x$-links increase efficiency more than active $y$-links do. Hence,  
\begin{equation}
		\JB(\alpha,S;G) > \JB(\alpha,S';G').
\end{equation}

Now, we consider the remaining case when $(S,G)$ induces $k_y > 1$ $y$-partitions $\{\mcal{P}_y^k\}_{k=1}^{k_y}$ and $k_x \geq 1$ $x$-partitions $\{\mcal{P}_x^k\}_{k=1}^{k_x}$. Consider $k_y$ such star subgraphs $\{\Gamma_y^k\}_{k=1}^{k_y}$ with center nodes $i_y^k$. Recast the $x$-partitions into similar star subgraphs $\{\Gamma_x^k\}_{k=1}^{k_x}$ with center nodes $i_x^k$. We first connect each $\Gamma_x^k$ to some $\Gamma_y^j$ with a single link $(i_x^k,i_y^j)$ in any manner as long as a link between the original $\mcal{P}_x^k$ and $\mcal{P}_y^j$ exists. For each excess outgoing edge, we create an isolated node with an $x$-imposter attached. Each isolated node is attached to a corresponding $i_y^k$ such that the original number of outgoing edges for each $\mcal{P}_y^k$ is satisfied. At this point, there are $k' \leq k_y$ connected components $G_k$ in the construction, and the efficiency of this construction is identical to the original. Lastly, we apply the efficiency reduction procedure from before for each $G_k$ to obtain  $(S_k',G_k') \in \mbb{S}_{m_k}$. From \eqref{eq:JI_partitions} and \eqref{eq:mediant}, we have
\begin{equation}
	\begin{aligned}
		\JB(\alpha,S;G) &> \frac{\sum_{k=1}^{k'} |e(\Gamma_y^k,\Gamma_y^k)| }{(1+\alpha_{\text{sys}})\sum_{k=1}^{k'} |e(\Gamma_y^k,\mcal{N})|} \\
		&\geq \min_{k=1,\ldots,k'} \JB(\alpha,S_k';G_k').
	\end{aligned}
\end{equation}
\end{proof}


\begin{proof}[Proof of Lemma \ref{planes_lemma}: Technical result for expected risks]
	Let us define $f(\bm{p}):= \sum_{i=1}^M p_i s_i$ and $f_i(\bm{p}) :=  \left( \sum_{j=1}^k p_j \right) (\nu_k-1) + 1$. The set of probability vectors such that $k = \argmin{i=1,\ldots,M} f_i(\bm{p})$ can be written as the set
	\begin{equation}\label{eq:min_regions}
		\begin{aligned}
			V_k &:= \left\{ \bm{p} \in \Delta_M : f_k(\bm{p}) \leq f_\ell(\bm{p}) \ \forall \ell \neq k \right\} \\
			&= \bigcap_{\ell \neq k} \left\{ \bm{p} \in \Delta_M : \frac{\sum_{j=1}^k p_j}{\sum_{j=1}^\ell p_j} \geq \frac{1-\nu_\ell}{1-\nu_k} \right\}.
		\end{aligned}
	\end{equation} 
	
	Define $\lambda_k = \frac{\sum_{j=1}^k d_j }{\sum_{j=1}^{k+1} d_j}$ for each $k=1,\ldots,M-1$. With some algebra, we can express each $s_i$ as
	\begin{equation}
		\begin{aligned}
			s_i = \left[ \! \sum_{j=1}^{M-i+1} \frac{n_j}{d_j}(1\!-\! \lambda_{j-1}) \left( \prod_{k=j}^{M-1} \lambda_k \right) \! \right] +  \left( 1\!-\!\! \prod_{j=M-i+1}^{M-1} \lambda_j \right) \\
		\end{aligned}
	\end{equation}
	Using the identities $\sum_{k=1}^M (1 \!-\! \lambda_{k-1})\left( \prod_{j=k}^{M-1} \lambda_j \right) = 1$ and $\sum_{k=1}^\ell (1\!-\! \lambda_{k-1})\left( \prod_{j=k}^{M-1} \lambda_j \right) = \prod_{j=\ell}^{M-1} \lambda_j$ for $\ell \!=\! 1,\ldots,M-1$, we obtain (omitting the algebraic steps)
	\begin{equation}
		\begin{aligned}
			f(\bm{p}) &= \sum_{i=1}^M (1\!-\!\lambda_{i-1}) \!\! \left(\prod_{j=i}^{M-1}\lambda_j\right)\!\!\left[ \left(\frac{n_i}{d_i}-1 \right) \!\! \left( \sum_{j=1}^{M-i+1} p_j\right) \!+\! 1\right] \\
		\end{aligned}
	\end{equation}
	Now, suppose $\bm{p} \in V_k$ for $k \in \{1,\ldots,M\}$.  Using \eqref{eq:min_regions} and $\nu_i \leq n_i/d_i \leq 1$, we then have
	\begin{equation}
		\begin{aligned}
			f(\bm{p}) &\geq \sum_{i=1}^M \lambda_{i-1}\left(\prod_{j=i}^{M-1}\lambda_j\right)f_k(\bm{p}) = f_k(\bm{p}) .
		\end{aligned}
	\end{equation}
\end{proof}







\bibliographystyle{IEEEtran}
\bibliography{sources}

\end{document}